\documentclass[11pt]{amsart}
\usepackage{amssymb,mathrsfs,graphicx,enumerate}
\usepackage{amsmath,amsfonts,amssymb,amscd,amsthm,bbm}
\allowdisplaybreaks
\usepackage[retainorgcmds]{IEEEtrantools}
\usepackage{pgfplots}
\usepackage{tikz}
\usepackage{tikz-3dplot}
\usepackage{sansmath}

\usetikzlibrary{shadings,intersections}
\usepackage{colortbl}

\title[Emergent dynamics of the Winfree sphere model]{Generalization of the Winfree model to the high-dimensional sphere and its emergent dynamics}
 
\author[Park]{Hansol Park}
\address[Hansol Park]{\newline Department of Mathematical Sciences\newline Seoul National University, Seoul 08826, Republic of Korea}
\email{hansol960612@snu.ac.kr}

\newtheorem{theorem}{Theorem}[section]
\newtheorem{lemma}{Lemma}[section]
\newtheorem{corollary}{Corollary}[section]
\newtheorem{proposition}{Proposition}[section]
\newtheorem{remark}{Remark}[section]

\newtheorem{definition}{Definition}[section]

\newcommand{\bbr}{\mathbb R}

\newcommand{\bbs}{\mathbb S}
\newcommand{\bbz}{\mathbb Z}

\makeatletter
\@namedef{subjclassname@2020}{\textup{2020} Mathematics Subject Classification}
\makeatother

\begin{document}

\date{\today}

 \keywords{}

\thanks{\textbf{Acknowledgment.} The work of H. Park is supported by National Research Foundation of Korea (NRF-2020R1A2C3A01003881).}

\begin{abstract}
We present a high-dimensional Winfree model in this paper. The Winfree model is the first mathematical model for synchronization on the circle. We generalize this model to the high-dimensional sphere and we call it ``the Winfree sphere model." We restricted the support of the influence function in the neighborhood of the attraction point with a small diameter to mimic the influence function as the Dirac-delta distribution. Restricting the support of the influence function allows several new conditions of the complete phase-locking states for the identical Winfree sphere model compare to previous results. We also provide the exponential $\ell^1$-stability and the existence of the equilibrium solution to obtain the complete oscillator death state of the Winfree sphere model. 
\end{abstract}

\maketitle \centerline{\date}
%\tableofcontents

\section{Introduction}\label{sec:1}
Self-organization of many body systems in the real worlds, such as the flocking of birds, schooling of fish \cite{Ao, D-M, T-B}, and flashing of fireflies \cite{B-C} are ubiquitous. Nowadays, they are getting a spotlight regarding unmanned aerial vehicles(UAV) \cite{A-B-F}, cooperative robot system \cite{D-B, D-F, G-K}, etc. Since two pioneers A. Winfree \cite{Winf} and Y. Kuramoto \cite{Kura1, Kura2} of the collective dynamics, works on this field of research have been studied intensively.

In this paper, we study the following system which is called the Winfree sphere model given as follows:
\begin{align}\label{A-0}2
\begin{aligned}
\begin{cases}
&\dot{x}_i=\Omega_ix_i+\displaystyle\frac{\kappa}{N}(e-\langle x_i, e\rangle x_i)\eta(x_i)\sum_{j=1}^NI(x_j)\quad t>0,\\
&x_i(0)=x_i^0\in\bbs^d,\quad \forall i\in\mathcal{N}:=\{ 1, 2, \cdots, N\},
\end{cases}
\end{aligned}
\end{align}
where $\kappa\in\bbr_+$ is the coupling strength, and $\Omega_i$ is a natural frequency matrix of the $i^{th}$ particle. Here, $e$ is the attraction point of the given system and it is a generalization of the points $2\pi n$ with $n\in\bbz$ in the Winfree model \eqref{A-1}. Since the Winfree coupling induce the attractive force between $x_i$ and $e$, we call $e$ an attraction point. The Winfree sphere model is a generalization of the Winfree model proposed by A. Winfree in 1967 given as follows:
\begin{align}\label{A-1}
\begin{cases}
\displaystyle\dot{\theta}_i=\nu_i+\frac{\kappa}{N}\sum_{k=1}^NS(\theta_i)I(\theta_k)\quad t>0,\\
\theta_i(0)=\theta_i^0\quad i\in \mathcal{N},
\end{cases}
\end{align}
where $\kappa\in\bbr_+$ is the coupling strength and $\nu_i\in\bbr$ is the natural frequency of the $i^{th}$ particle. Here, $S$ and $I$ are the sensitivity and influence functions, respectively. We refer to the swarm sphere model \cite{O1, S-H} to construct the Winfree-coupling and the free flow on the sphere. We present the relation between two models \eqref{A-0} and \eqref{A-1} in Section \ref{sec:3}.

The Winfree model \eqref{A-1} with the specific functions $(S(\theta), I(\theta))=(-\sin\theta, 1+\cos\theta)$ has been studied in previous works \cite{A-S, N-M-P, Q-R-S1, Q-R-S2}. In \cite{A-S}, phase diagram for the Winfree model has been studied. Several variants of the Winfree model have been studied. Dor example, a time-delay effect and general network topology \cite{H-K}, a locally coupled network topology \cite{H-K-P-R-(2)}, a partial phase-locking state \cite{H-K-P-R-(3)}, and the kinetic analogue \cite{H-P-Z}. An infinite cylinder $\mathbb{T}\times\mathbb{R}$ extension of the Winfree model \cite{H-K-M}.

Furthermore, authors of \cite{H-P-R} generalized these paper to general functions $(S, I)$ under reasonable conditions. We denote these conditions in Theorem \ref{thm2.1}. This conditions of $(S, I)$ focused on the geometric structure of the graphs. However, in this paper, we approach to influence function from a different point of view. Since we want to put an approximated delta function into the influence function $I$, in this paper we assume that the support of the influence function is contained in the neighborhood of the attraction point. The framework of the influence function $I$ is introduced in Section \ref{sec:3.2}. We assume that $\mathrm{supp}I=\mathcal{D}_\beta:=\{x\in\bbs^d: \angle(e, x)\leq \beta\}$. Here $\angle(a, b):=\arccos\langle a, b\rangle$ is an angle between two vectors $a, b\in\bbs^d$. This assumption is quite new approach, since support of a usually considered influence function $I(\theta)=1+\cos\theta$ is $\bbr$.  Also, the authors of \cite{H-P-R} only considered the differentiable influence function. In this paper, we only assumed the Lipschitz continuity of $I$. \\

The main results of this paper are three-fold. First, we construct the Winfree sphere model and introduce a framework of the influence function $I$. 

Second, we study emergent dynamics of the Winfree sphere model. Here, we study invariant set of the Winfree sphere model. If coupling strength is sufficiently large to satisfy the following inequality for a $0<\gamma<\frac{\pi}{2}$:
\[
\sin\gamma\tilde{I}(\gamma)>\frac{1}{\kappa}\max_{j\in\mathcal{N}}\|\Omega_j\|_{\mathrm{op}},
\]
then $\mathcal{D}_\gamma=\{x\in\bbs^d: \angle(e, x)\leq \gamma\}$ is an invariant set of the Winfree sphere model. Here $\|\cdot\|_{\mathrm{op}}$ is the operator norm(see Proposition \ref{prop4.1}). We also show that if a coupling strength is sufficiently large compare to $\max_{j\in\mathcal{N}}\|\Omega_j\|_{\mathrm{op}}$, then every particles move into $\mathcal{D}_\beta$ in finite time $T$ and stay in $\mathcal{D}_\beta$. We could express the explicit value of $T=\max_{j\in\mathcal{N}}T_j$ with $T_j$ given in \eqref{D-10}(see Theorem \ref{thm4.1} and Remark \ref{rmk4.2}). 

Third, for the identical Winfree sphere model \eqref{E-1}, if a coupling strength is large enough than operator norm of the common natural frequency $\|\Omega\|_{\mathrm{op}}$, then exponential aggregation occurs(see Theorem \ref{thm5.2}). This result is a simple generalization of previous results into the Winfree sphere model.  For the arbitrary coupling strength $\kappa>0$, we have a weaker result. From the assumption of the influence function introduced in Subsection \ref{sec:3.2}, if the initial data satisfy $\angle(x_i, x_j)<\frac{\pi}{2}-\beta$ for all $i, j\in\mathcal{N}$, then we could obtain a dichotomy for the long-time behaviors:
\begin{enumerate}
\item Complete synchronization occurs. i.e.,
\[
\lim_{t\to\infty}\|x_c\|=1.
\]

\item All particles escape the support of the influence function $I$:
\[
\lim_{t\to\infty}I_c(t)=0\quad\Leftrightarrow \quad \liminf_{t\to\infty}\angle(e, x_i)\geq\beta\quad\forall~i\in\mathcal{N}.
\]
Furthermore, all particles move along the free flow after a sufficiently large time $T>0$:
\[
\dot{x}_i=\Omega x_i\quad \forall ~i\in\mathcal{N},~ t\geq T.
\]
\end{enumerate}
Detail proof and statement are provided in Theorem \ref{thm5.3}. For the original Winfree model \eqref{A-1}, if $\dot{\theta}_i=\nu$ with $\nu\neq0$, then particles meet attraction point $e$ infinitely many times, so the second case is impossible for the original Winfree model.  \\

The rest of this paper is organized as follows. In Section \ref{sec:2}, we introduce previous results of the Winfree model and basic lemma that we use in this paper. In Section \ref{sec:3}, we construct the Winfree model on the high-dimensional sphere($\bbs^d$), and we call this model the Winfree sphere model. We also imposed some reasonable assumptions on the influence function $I$. In Section \ref{sec:4}, we provide the emergent dynamics of the Winfree sphere model. In Section \ref{sec:5}, we provided the aggregation of the Winfree sphere model if all particles have the same natural frequency. Finally, Section \ref{sec:6} is devoted to a brief summary of the paper.

\setcounter{equation}{0}
\section{Preliminaries}\label{sec:2}

In this section, we provide previous results of the original Winfree model in \eqref{A-1}. Below, we recall a definition of the asymptotic phase-locked state and the complete synchronization.
\begin{definition}\label{def2.1}
Let $\Theta=(\theta_1, \theta_2, \cdots, \theta_N)$ be a solution of the Winfree model \eqref{A-1}.
\begin{enumerate}
\item The solution $\Theta=\Theta(t)$ tends to the (strong) phase-locked state asymptotically if and only if its transversal phase differences tend to constant values: for all $i, j\in\mathcal{N}$, there exists a constant value $\theta_{ij}^\infty$, such that
\[
\lim_{t\to\infty}|\theta_i(t)-\theta_j(t)|=\theta_{ij}^\infty.
\]
\item The solution $\Theta=\Theta(t)$ exhibits a complete (frequency) synchronization if and only if its transversal frequency differences tend to zero: for all $i, j\in\mathcal{N}$, we have
\[
\lim_{t\to\infty}|\dot{\theta}_i(t)-\dot{\theta}_j(t)|=0.
\]
\end{enumerate}
\end{definition}
We can generalize the above definition to the Winfree sphere model as follows.
\begin{definition}\label{def2.2}
Let $\mathcal{X}=(x_1, x_2, \cdots, x_N)$ be a solution of the Winfree sphere model \eqref{A-0}. The solution $\mathcal{X}=\mathcal{X}(t)$ tends to the (strong) phase-locked state asymptotically if and only if its transversal phase differences tend to constant values: for all $i, j\in\mathcal{N}$, there exists a constant value $L_{ij}$, such that
\[
\lim_{t\to\infty}|x_i(t)-x_j(t)|=L_{ij}.
\]
\end{definition}
Before we present some results, we assume that sensitivity and influence functions satisfy the following conditions in the rest of this section.\vspace{0.2cm}
\begin{itemize}
\item $(\mathcal{A}1)$: the sensitivity function $S$ is a $2\pi$-periodic, $\mathcal{C}^2$, and odd function; and the influence function $I$ is a $2\pi$-periodic, $\mathcal{C}^2$, and even function.\vspace{0.2cm}

\item $(\mathcal{A}2)$: the sensitivity and influence functions satisfy some geometric conditions: there exist a positive constants $\theta_*$ and $\theta^*$, satisfying
\[
0<\theta_*<\theta^*<2\pi,
\]
such that,
\begin{enumerate}
\item $S\leq 0$ on $[0, \theta^*]$ and $S'\leq 0$, $S''\geq0$ on $[0, \theta_*]$,
\item $I\geq0$, $I'\leq 0$ on $[0, \theta^*]$, and $I''\leq0$ on $[0, \theta_*]$, 
\item $(SI)'<0$ on $(0, \theta_*)$, $(SI)'>0$ on $(\theta_*, \theta^*)$.
\end{enumerate}
\end{itemize}

From the above assumption,  we set $\alpha^\infty$ is the unique solution of $(SI)(x)=(SI)(\alpha)$, $x\in[0, \theta_*]$ for given $\alpha\in(0, \theta^*)$. We also define 
\begin{align*}
\begin{cases}
\displaystyle K_e(\alpha^\infty):=-\frac{\max_{j\in\mathcal{N}}|\nu_i|}{S(\alpha^\infty)I(\alpha^\infty)},\\
\mathcal{R}(\alpha):=\{\Theta=(\theta_1, \cdots, \theta_N)\in\bbr^N: \theta_i\in(-\alpha, \alpha), ~\forall~i\in\mathcal{N}\}.
\end{cases}
\end{align*}

The following theorem is about the existence of the phase-locked state.
\begin{theorem}[\cite{H-P-R} Existence of the phase-locked state]\label{thm2.1}
Suppose that the sensitivity and influence functions satisfy $(\mathcal{A}1)$-$(\mathcal{A}2)$, and initial data satisfy
\[
\Theta_0\in \overline{\mathcal{R}(\alpha)}\quad\text{and}\quad \kappa>\kappa_e(\alpha^\infty).
\]
Let $\Theta=\Theta(t)$ be a solution of the original Winfree model \eqref{A-1}. Then, $\Theta(t)$ converges to a unique equilibrium state $\Phi=(\phi_1, \cdots, \phi_N)$ in the region $\mathcal{R}(\alpha^\infty)$, i.e., there exists a unique phase-locked state $\Phi:=(\phi_1,\cdots, \phi_N)\in\mathcal{R}(\alpha^\infty)$, such that
\[
\nu_i+\frac{\kappa}{N}S(\phi_i)\sum_{j=1}^NI(\phi_j)=0,\quad\lim_{t\to\infty}\Theta(t)=\Phi.
\]
\end{theorem}

\begin{proof}[Sketch of proof]
For a detail proof, we refer to Theorem 2.2 of \cite{H-P-R}. We only provide a sketch of the proof. First, we show any solution $\Theta(t)$ enters $\mathcal{R}(\alpha^\infty)$ within some finite time. Second, using the stability estimate in $\ell^1$-distance, we show any two solutions to \eqref{A-1} with initial datas in $\mathcal{R}(\alpha^\infty)$ must converges to each other asymptotically. Third, we show the existence of an equilibrium in $\mathcal{R}(\alpha^\infty)$. Finally, combining these three facts yields, any solution $\Theta(t)$ converges to equilibrium asymptotically.
\end{proof}

Now we introduce the results of the positive invariance of the Winfree model.
\begin{lemma}[\cite{H-P-R} Existence of positively invariant set]\label{lem2.1}
 Suppose that $\Theta=\Theta(t)$ be a solution introduced in Theorem \ref{thm2.1}. Let $\alpha$ and the coupling strength satisfy
\[
\alpha\in(0, \theta^*)\quad\text{and}\quad \kappa>\kappa_e(\alpha^\infty).
\]
Then, the set $\mathcal{R}(\alpha^\infty)$ is positively invariant along the Winfree flow \eqref{A-1}:
\[
\Theta_0\in\mathcal{R}(\alpha^\infty)\quad\Rightarrow \Theta(t)\in\mathcal{R}(\alpha^\infty),\quad t\in(0, \infty).
\]
\end{lemma}

\begin{proof}
We refer to Lemma 3.1 of \cite{H-P-R}.
\end{proof}
We provide the Winfree sphere model version of Lemma \ref{lem2.1} in Proposition \ref{prop4.1}. We also introduce the results on the transition to invariant set of the Winfree model.

\begin{proposition}[\cite{H-P-R} Transition to $\mathcal{R}(\alpha^\infty)$]\label{prop2.1}
Suppose that $\Theta=\Theta(t)$ be a solution introduced in Theorem \ref{thm2.1}.  Let $\alpha$ and the coupling strength satisfy
\[
\alpha\in(0, \theta^*)\quad\text{and}\quad \kappa>\kappa_e(\alpha^\infty).
\]
Let $\Theta=(\theta_1, \cdots, \theta_N)$ be a global smooth solution to \eqref{A-1}, satisfying $\Theta_0\in\overline{\mathcal{R}(\alpha)}$, then there exists $t_*\geq0$ such that
\[
\Theta(t)\in\mathcal{R}(\alpha^\infty),\quad \forall ~t>t_*.
\]
\end{proposition}
\begin{proof}
We refer to Proposition 1 of \cite{H-P-R}.
\end{proof}
We provide the Winfree sphere model version of Proposition \ref{prop2.1} in Theorem \ref{thm4.1}. Now we present the following Barbalat's lemma without a proof.
\begin{lemma}[\cite{Ba} Barbalat's lemma]\label{lem2.2}
Suppose that  a real-valued function $f: [0, \infty) \to \bbr$ is uniformly continuous and it satisfies
\[ \lim_{t \to \infty} \int_0^t f(s)d s \quad \textup{exists}. \]
Then, $f$ tends to zero as $t \to \infty$:
\[ \lim_{t \to \infty} f(t) = 0. \]
\end{lemma}

\setcounter{equation}{0}
\section{Generalization of the Winfree model}\label{sec:3}
In this section, we generalize the Winfree model \eqref{A-1} to the high-dimensional sphere($\bbs^d$) with $d\geq2$ to obtain the Winfree sphere model \eqref{A-0}. 

\subsection{Construction of the Winfree sphere model}
The original Winfree model was defined on the circle($\mathbb{S}^1$) and the system is given as follows:
\[
\dot{\theta}_i=\nu_i+\frac{\kappa}{N}\sum_{j=1}^N S(\theta_i)I(\theta_j)=\nu_i+\kappa S(\theta_i) I_c,
\]
where for the handy notation, we denote the average of $I(\theta_k)$ as $I_c$. i.e. $I_c=\frac{1}{N}\sum_{k=1}^NI(\theta_k)$. Let $x_1, x_2, \cdots, x_N\in\bbs^d\subset\bbr^{d+1}$. Since we can consider that each particle $\theta_i$ in the Winfree model was attracted by the special point $\theta=0$, we define this kind of special point $e$ in $\bbs^d$ which attracting particles as follows:
\[
e=[1, \underbrace{0, 0, \cdots, 0}_{d\text{-times}}]^{\top}\in\bbs^d\subset\bbr^{d+1}.
\]
Throughout this paper, we call $e$ \textit{the attraction point} of the Winfree sphere model. If we assume that the velocity of the $i^{th}$ particle is heading to the attraction point, then the velocity of the $i^{th}$ particle can be expressed as
\[
\dot{x}_i=\alpha_i\frac{e-\langle x_i, e\rangle x_i}{\|e-\langle x_i, e\rangle x_i\|}
\]
with some non-negative constant $\alpha_i$. Recall that $\|e-\langle x_i, e\rangle x_i\|=\sqrt{1-\langle x_i, e\rangle^2}$. This means, if $e=x_i$, then $\dot{x}_i$ is not well defined if $\alpha_i\neq0$ at $x_i=e$. For the further condition, we assume that $\alpha_i=0$ if $x_i=e$ for the well-definedness of the system. Winfree assumed that the velocity of $x_i$ heading to $e$ is influenced from other $j^{th}$ particles $x_j$ independently, we can write as follows:
\[
\dot{x}_i=\sum_{j=1}^N \beta_{ij}\frac{e-\langle x_i, e\rangle x_i}{\|e-\langle x_i, e\rangle x_i\|}.
\]
We also know that $\beta_{ij}$ is proportional to $S_i=S(x_i)$ (the sensitivity of the $i^{th}$ particle), $I_j=I(x_j)$ (the influence of the $j^{th}$ particle), and coupling strength $\kappa$ in the original Winfree model. Following this philosophy, we assume
\[
\beta_{ij}=\kappa_iS(x_i)I(x_j).
\]
Now we want to construct the natural frequency term of the high-dimensional Winfree model. To generalize the natural frequency terms of the Winfree model to the natural frequency terms of the Winfree sphere model, we refer to a relation between the Kuramoto model \cite{A-B, A-R, D-X, H-K-R-(1), J-C} and the Lohe sphere model \cite{H-K-P-R, Lo-0, Lo-1, Lo-2}. Then the canonical extension of the natural frequency term to the Winfree sphere model can be expressed as follows:
\[
\dot{x}_i=\Omega_ix_i,
\]
where $\Omega_i\in\mathrm{Skew}_{d+1}\bbr$. Here, we denote a set of skew-symmetric matrices with size $(d+1)\times (d+1)$ as $\mathrm{Skew}_{d+1}\bbr$. On the other hand, we can interpret this natural frequency term as a natural frequency tensor $A_i$ introduced in the Lohe tensor model \cite{H-P-(1)}. Natural frequency tensors of the Lohe tensor model was proposed to preserve the norm of given tensors along the time evolution. Since we want to preserve norm of vectors, definition of natural frequency term for the Winfree sphere model is natural.

Based on the above arguments, we obtain the high-dimensional Winfree model as follows:
\begin{align}\label{C-2}
\begin{cases}
\dot{x}_i=\displaystyle\Omega_ix_i+\frac{\kappa}{N}\sum_{j=1}^NS(x_i)I(x_j)\frac{e-\langle x_i, e\rangle x_i}{\sqrt{1-\langle x_i, e\rangle^2}},\quad t>0,\\
x_i(0)=x_i^0\in\bbs^d,\quad \forall ~i\in\mathcal{N}.
\end{cases}
\end{align}
As we mentioned before, system \eqref{C-2} does not well-defined at $x_i=e$. So we want to assume that $S(x_i)I(x_j)=0$ if $x_i=e$. Of course, $I(x_j)$ has no information of $x_i$, it cannot be 0 when $x_i=0$. So we define $S(e)=0$. To make a well-defined system, we assume that $S(x_i)$ can be expressed as following form:
\begin{align}\label{C-3}
S(x_i)=\eta(x_i)\sqrt{1-\langle x_i, e\rangle^2}
\end{align}
for some Lipschitz continuous function $\eta$ defined on $\bbs^d$. Now we substitute \eqref{C-3} into \eqref{C-2} to obtain the following system:
\begin{align}
\begin{aligned}\label{C-4}
\begin{cases}
&\dot{x}_i=\Omega_ix_i+\displaystyle\frac{\kappa}{N}(e-\langle x_i, e\rangle x_i)\eta(x_i)\sum_{j=1}^NI(x_j),\quad t>0,\\
&x_i(0)=x_i^0\in\bbs^d,\quad \forall i\in\mathcal{N}.
\end{cases}
\end{aligned}
\end{align}
We call this model as ``\textit{the Winfree sphere model.}"\\

Next, we briefly discuss how to reduce the Winfree sphere model \eqref{A-0} to the original Winfree model \eqref{A-1}. First we set $d=1$. Then, attraction point $e$, each particle $x_i$ and natural frequency $\Omega_i$ can be expressed as follows:
\[
e=\begin{bmatrix}
1\\
0
\end{bmatrix},\quad
x_i=\begin{bmatrix}
\cos\theta_i\\
\sin\theta_i
\end{bmatrix},\quad \Omega_i=\begin{bmatrix}
0&-\nu_i\\
\nu_i&0
\end{bmatrix}
\]
for some real numbers $\theta_i$ and $\nu_i$. Then \eqref{C-4}$_1$ can be expressed as
\begin{align*}
\dot{\theta}_i\begin{bmatrix}
-\sin\theta_i\\
\cos\theta_i
\end{bmatrix}=\nu_i\begin{bmatrix}
-\sin\theta_i\\
\cos\theta_i
\end{bmatrix}+\frac{\kappa}{N}\left(
\begin{bmatrix}
1\\
0
\end{bmatrix}-\cos\theta_i \begin{bmatrix}
\cos\theta_i\\
\sin\theta_i
\end{bmatrix}
\right)\eta(\theta_i)\sum_{j=1}^N I(\theta_j).
\end{align*}
This can be simplified as follows:
\[
\dot{\theta}_i=\nu_i-\frac{\kappa}{N}\sin\theta_i\eta(\theta_i)\sum_{j=1}^NI(\theta_j).
\]
Finally, we put $S(\theta_i)=-\sin\theta_i \eta(\theta_i)$ to obtain \eqref{A-1}$_1$. So the Winfree sphere model \eqref{C-4} is a generalization of the original Winfree model \eqref{A-1}.\\

Now we have the following simple lemma to show that each particles $\{x_i(t)\}$ lie on the unit sphere $\bbs^d$.
\begin{lemma}
Let $\{x_i\}$ be a solution of system \eqref{C-4}. Then we have
\[
\|x_i(t)\|=1\quad\forall t\geq0,\quad i\in\mathcal{N}.
\]
\end{lemma}

\begin{proof}
By direct calculations, we have
\begin{align*}
\frac{d}{dt}\langle x_i, x_i\rangle&=2\langle x_i,\dot{x}_i\rangle
=2\left\langle x_i, \Omega_i x_i+\frac{\kappa}{N}\sum_{j=1}^N(e-\langle x_i, e\rangle x_i)\eta(x_i)I(x_j)\right\rangle=0.
\end{align*}
Here, we used the facts $\langle x_i, \Omega_i x_i\rangle=0$ and $\langle x_i, e-\langle x_i, e\rangle x_i\rangle=0$. Finally, we can obtain that
\[
\|x_i(t)\|^2=\|x_i(0)\|^2=1
\]
for all $t\geq0$.
\end{proof}
From now on, we consider only the case when $\eta\equiv 1$. If we substitute $\eta\equiv 1$ into \eqref{C-4}, we have the following system:
\begin{align}
\begin{aligned}\label{C-5}
\begin{cases}
&\dot{x}_i=\Omega_ix_i+\displaystyle\frac{\kappa}{N}(e-\langle x_i, e\rangle x_i)\sum_{j=1}^NI(x_j)\quad t>0,\\
&x_i(0)=x_i^0\in\bbs^d,\quad \forall i\in \mathcal{N}.
\end{cases}
\end{aligned}
\end{align}

\subsection{A framework of the influence function $I$.}\label{sec:3.2}

In this subsection, we introduce a framework of the influence function $I$. Here we consider how to approximate the Dirac-delta distribution.\\

We define the angle between two vectors $x, y\in\bbs^d$ as follows:
\[
\angle(x, y)=\arccos\langle x, y\rangle.
\]
This angle($\angle$) can be considered as the length of a shortest geodesic path between $x$ and $y$ on the unit sphere $\bbs^d$. From this reason, the angle function is a metric on $\bbs^d$, so we have the triangle inequality expressed as follows:
\begin{align}\label{C-5-1}
\angle(x, y)+\angle(y, z)\geq \angle(x, z)\quad \forall ~x, y, z\in\bbs^d.
\end{align}

Since the influence function $I$ is a relaxed smooth function of the Dirac-delta distribution $\delta(x-e)$, we can assume that $I$ is positive and the support of $I$ is a subset of $\mathcal{D}_{\frac{\pi}{2}-\alpha}$ for some $0<\alpha<\frac{\pi}{2}$, where $\mathcal{D}_\theta=\left\{x\in\bbs^d: \angle( e, x)\leq \theta\right\}$. i.e., $I(x)\geq0$ for all $x\in\bbs^d$ and $I(x)=0$ if $\angle(e, x)\geq\frac{\pi}{2}-\alpha$ for the fixed $0<\alpha\leq \frac{\pi}{2}$. We can draw $\mathcal{D}_{\frac{\pi}{2}-\alpha}$ on the unit sphere as Figure \ref{Fig1}.

\begin{center}
\begin{figure}[h]
\begin{tikzpicture}
  \coordinate (O) at (0,0);

  % ball background color
  \shade[ball color = white, opacity = 0.2] (0,0) circle [radius = 2cm];

  % cone
  \begin{scope}
    \def\rx{0.71}% horizontal radius of the ellipse
    \def\ry{0.15}% vertical radius of the ellipse
    \def\z{0.725}% distance from center of ellipse to origin

    \path [name path = ellipse]    (0,\z) ellipse ({\rx} and {\ry});
    \path [name path = horizontal] (-\rx,\z-\ry*\ry/\z)
                                -- (\rx,\z-\ry*\ry/\z);
    \path [name intersections = {of = ellipse and horizontal}];
  \end{scope}
  % label of cone
  %\draw (0.25,0.4) -- (0.9,0.1) node at (1.05,0.0) {$q$};
  \draw[thick] (0.5,0) arc (0:45:0.5) node at (0.7, 0.3) {$\alpha$};
  % ball
  \draw (O) circle [radius=2cm];
  % label of ball center point
  \filldraw (O) circle (1pt) node[below] {$O$};
  \filldraw (0, 2) circle (1pt) node [above] {$e$};
  % radius
  \draw[densely dashed] (O) --(-1.33,1.33);
  \draw [densely dashed](O) -- (1.33,1.33);
  %\draw (O) -- (0, 2);
  \draw [densely dashed](-2, 0) -- (2, 0);
  % cut of ball surface
  \draw (-1.35,1.47) arc [start angle = 140, end angle = 40,
    x radius = 17.6mm, y radius = 14.75mm];
  \draw[densely dashed] (-1.36,1.46) arc [start angle = 170, end angle = 10,
    x radius = 13.8mm, y radius = 3.6mm];
  \draw (-1.29,1.52) arc [start angle=-200, end angle = 20,
    x radius = 13.75mm, y radius = 3.15mm];
    \draw (-2,0) arc [start angle=-180, end angle = 0,
    x radius = 20mm, y radius = 6mm];
    \draw[densely dashed] (-2,0) arc [start angle=180, end angle = 0,
    x radius = 20mm, y radius = 6mm];
  % label of cut of ball surface
  \draw[->] (-1,2) -- (-0,1.5) node at (-1.2,2.2) {$\mathcal{D}_{\frac{\pi}{2}-\alpha}$};
\end{tikzpicture}
\caption{Geometric meaning of $\mathcal{D}_{\frac{\pi}{2}-\alpha}$ when $d=2$.}
\label{Fig1}
\end{figure}
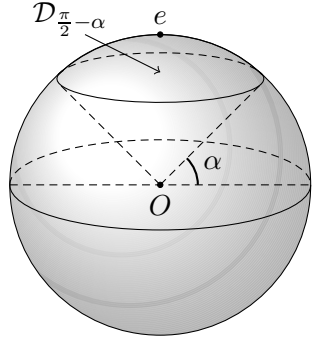
\end{center}
The Winfree sphere model is defined on the sphere $\bbs^d$, so we assume the symmetry of $I$:
\[
\angle(e, x)=\angle(e, y)\quad\Longrightarrow \quad I(x)=I(y).
\]
This implies that there exists $\tilde{I}: [0, \pi]\to [0, \infty)$ such that
\[
I(x)=\tilde{I}(\angle(e, x))\quad\forall ~x\in\bbs^d.
\]
Since $\mathrm{supp}I\subseteq \mathcal{D}_{\frac{\pi}{2}-\alpha}$, we have $\mathrm{supp}\tilde{I}\subseteq \left[0, \frac{\pi}{2}-\alpha\right]$. Physically, closer particle makes more influence. In the work of $I$, if $x$ is closer to $e$, then $I(x)$ should be larger. From this reason, we assume that $\tilde{I}$ is a decreasing function. Since $\tilde{I}$ is decreasing, the support of $\tilde{I}$ can be expressed as a closed interval $[0, \beta]$ for some $0<\beta\leq \frac{\pi}{2}-\alpha$. Finally, we assume the Lipschitz continuity of $\tilde{I}$ to guarantee the uniqueness and the existence of a solution for system \eqref{C-5}. We can summarize conditions of $\tilde{I}:[0, \pi]\to [0, \infty)$ as follows:

\noindent$(\mathcal{C}1)$: $\tilde{I}$ is a decreasing function.

\noindent$(\mathcal{C}2)$: the support of $\tilde{I}$ forms $[0, \beta]$ for some $0<\beta\leq \frac{\pi}{2}-\alpha$.

\noindent$(\mathcal{C}3)$: the maximum of $\tilde{I}$ is normalized as $\tilde{I}(0)=1$.

\noindent$(\mathcal{C}4)$: $\tilde{I}$ is Lipschitz continuous.\\

From now on, we define $0<\beta<\frac{\pi}{2}$ which satisfies the follows:
\[
\mathrm{supp}{\tilde{I}}=[0, \beta],
\]
and we use this definition throughout the whole paper. For example, we consider two functions $\tilde{I}_1$ and $\tilde{I}_2$ given as follows:
\begin{align*}
\tilde{I}_1(x)=\begin{cases}
1-x,\quad&x\in[0, 1]\\
0,&x\in[1, \pi]
\end{cases},\quad
\tilde{I}_2(x)=\begin{cases}
(2x-1)^2,\quad&x\in[0, 0.5]\\
0,&x\in[0.5, \pi]
\end{cases}.
\end{align*}
The graph of them are plotted in Figure \ref{Fig2}. We know that $\tilde{I}_1$ and $\tilde{I}_2$ are decreasing,
\[
\begin{cases}
\mathrm{supp}\tilde{I}_1=[0, 1],\quad \beta_1=1,\quad \beta_1< \frac{\pi}{2},\\
\mathrm{supp}\tilde{I}_2=[0, 0.5],\quad \beta_2=0.5,\quad \beta_2< \frac{\pi}{2}, 
\end{cases}
\]
$\tilde{I}_1(0)=\tilde{I}_2(0)=1$, and of course they are Lipschitz continuous. So $\tilde{I}_1$ and $\tilde{I}_2$ satisfy conditions $(\mathcal{C}1)$, $(\mathcal{C}2)$, $(\mathcal{C}3)$, and $(\mathcal{C}4)$.

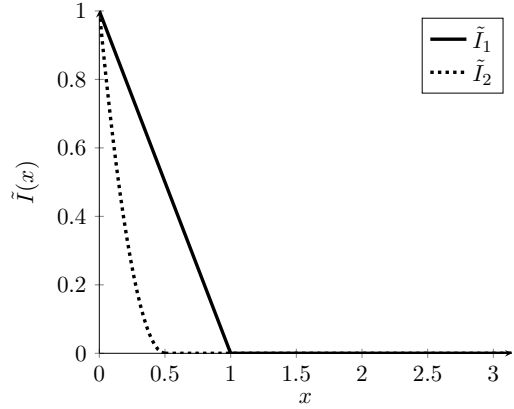
\begin{figure}[h]
\begin{tikzpicture}[scale = 0.8]
\begin{axis}[
    axis lines = left,
    xlabel = $x$,
    ylabel = {$\tilde{I}(x)$},
    every axis plot/.append style={ultra thick}
]
\addplot [
    domain=0:1,
    samples=3, 
]
{1-x};
\addlegendentry{$\tilde{I}_1$}
\addplot [
    domain=0:0.5, 
    samples=20, 
    dotted,
    ]
    {4*(x-0.5)^2};
\addplot  [
    domain=0.5:3.14, 
    samples=2, 
    dotted,
    ]
    {0};
    \addplot  [
    domain=1:3.14, 
    samples=2,
    ]
    {0};
\addlegendentry{$\tilde{I}_2$}
\end{axis}
\end{tikzpicture}
\caption{Examples of $\tilde{I}$.}
\label{Fig2}
\end{figure}

From the simple calculation, we have the following relation:
\begin{align*}
\mathrm{Lip}I&=\sup_{x\neq y}\frac{|I(x)-I(y)|}{\|x-y\|}=\sup_{x\in\bbs^d}\limsup_{y\to x}\frac{|I(x)-I(y)|}{\|x-y\|}\\
&=\sup_{x\in\bbs^d}\limsup_{y\to x}\left(\frac{|\tilde{I}(\angle(x, e))-\tilde{I}(\angle(y, e))|}{|\angle(x, e)-\angle(y, e)|}\cdot\frac{|\angle(x, e)-\angle(y, e)|}{\|x-y\|}\right)\\
&=\sup_{x\in\bbs^d}\limsup_{y\to x}\left(\frac{|\tilde{I}(\angle(x, e))-\tilde{I}(\angle(y, e))|}{|\angle(x, e)-\angle(y, e)|}\right)
=\mathrm{Lip}\tilde{I}.
\end{align*}
This implies that if $\tilde{I}$ is a Lipschitz continuous function, then $I$ is also a Lipschitz continuous function, and they have the same Lipschitz constant.

\setcounter{equation}{0}
\section{Emergent dynamics of the Winfree sphere model}\label{sec:4}
In this section, we study the emergent dynamics of the Winfree sphere model \eqref{C-5}.

\subsection{Existence of an Invariant set}
In this subsection, we study some invariant sets of the Winfree sphere model when a sufficiently large coupling strength $\kappa$ is given. We have the following simple lemma.
\begin{lemma}\label{lemma4.1}
Let $\{x_i\}$ be a solution of the Winfree sphere model \eqref{C-5}. Then we have the following inequality:
\begin{align}\label{D-3}
\frac{d}{dt}\angle(x_j, e)\leq \|\Omega_j\|_{\mathrm{op}}-\frac{\kappa}{N}\sin\angle(x_j, e)\sum_{k=1}^N\tilde{I}(\angle(x_k, e)).
\end{align}

\end{lemma}

\begin{proof}
Assume that $\{x_i\}_{i=1}^N\subset\mathcal{D}_\gamma$ and $x_j\in\partial\mathcal{D}_\gamma$. We define a unit vector $n_e(x)$ for all $x\in\bbs^d$ as follows:
\[
n_e(x)=\frac{e-\langle x, e\rangle x}{\|e-\langle x, e\rangle x\|}.
\]
We provide a geometric meaning of $n_e(x)$ in Figure \ref{Fig3}. Then the definition of $n_e(x)$ yields
\[
-\frac{d}{dt}\angle(x_j, e)=\frac{\frac{d}{dt}\cos(\angle(x_j, e))}{\sin\angle(x_j, e)}=\frac{\frac{d}{dt}\langle x_j, e\rangle}{\sqrt{1-\langle x_j, e\rangle^2}}=\frac{\langle \dot{x}_j, e-\langle x_j, e\rangle x_j\rangle}{\sqrt{1-\langle x_j, e\rangle^2}}=\langle \dot{x}_j, n_e(x_j)\rangle.
\]
By simple calculations, we have
\begin{align*}
-\frac{d}{dt}\angle(x_j, e)&=\langle n_e(x_j), \dot{x}_j\rangle \\
&=\left\langle n_e(x_j), \Omega_jx_j+\frac{\kappa}{N}(e-\langle x_j, e\rangle x_j)\sum_{k=1}^NI(x_k)\right\rangle\\
&=\langle n_e(x_j), \Omega_j x_j\rangle+\frac{\kappa}{N}\|e-\langle x_j,e\rangle x_j\|\sum_{k=1}^N I(x_k)\\
&\geq-\|\Omega_j\|_{\mathrm{op}}+\frac{\kappa}{N}\sin\angle(x_j, e)\sum_{k=1}^N\tilde{I}(\angle(x_k, e)).
\end{align*}
This gives a desired result.
\end{proof}

\begin{center}
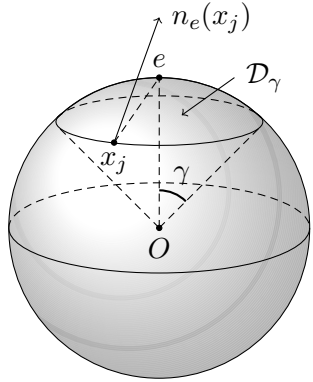
\begin{figure}[h]
\begin{tikzpicture}
  \coordinate (O) at (0,0);
  % ball background color
  \shade[ball color = white, opacity = 0.2] (0,0) circle [radius = 2cm];

  % cone
  \begin{scope}
    \def\rx{0.71}% horizontal radius of the ellipse
    \def\ry{0.15}% vertical radius of the ellipse
    \def\z{0.725}% distance from center of ellipse to origin

    \path [name path = ellipse]    (0,\z) ellipse ({\rx} and {\ry});
    \path [name path = horizontal] (-\rx,\z-\ry*\ry/\z)
                                -- (\rx,\z-\ry*\ry/\z);
    \path [name intersections = {of = ellipse and horizontal}];
  \end{scope}
  % label of cone
  %\draw (0.25,0.4) -- (0.9,0.1) node at (1.05,0.0) {$q$};
  \draw[thick] (0.3535,0.3535) arc (45:90:0.5) node at (0.3, 0.7) {$\gamma$};
  % ball
  \draw (O) circle [radius=2cm];
  % label of ball center point
  \filldraw (O) circle (1pt) node[below] {$O$};
  \filldraw (0, 2) circle (1pt) node [above] {$e$};
  % radius
  \draw[densely dashed] (O) --(-1.33,1.33);
  \draw [densely dashed](O) -- (1.33,1.33);
  %\draw (O) -- (0, 2);
  \draw [densely dashed](O) -- (0, 2);
  % cut of ball surfac
  \draw (-1.35,1.47) arc [start angle = 140, end angle = 40,
    x radius = 17.6mm, y radius = 14.75mm];
  \draw[densely dashed] (-1.36,1.46) arc [start angle = 170, end angle = 10,
    x radius = 13.8mm, y radius = 3.6mm];
  \draw (-1.29,1.52) arc [start angle=-200, end angle = 20,
    x radius = 13.75mm, y radius = 3.15mm];
    \draw (-2,0) arc [start angle=-180, end angle = 0,
    x radius = 20mm, y radius = 6mm];
    \draw[densely dashed] (-2,0) arc [start angle=180, end angle = 0,
    x radius = 20mm, y radius = 6mm];
  % label of cut of ball surface
  \draw[->] (1,2) -- (0.3,1.5) node at (1.4,2) {$\mathcal{D}_{\gamma}$};
  \filldraw (-0.6, 1.135) circle (1pt) node[below] {$x_j$};
  \draw[densely dashed] (-0.6, 1.135)--(0, 2);
  \draw[->] (-0.6, 1.135)--(0, 2.8) node at (0.7, 2.8) {$n_e(x_j)$};
\end{tikzpicture}
\caption{Geometric meaning of $n_e(x_j)$ in Lemma \ref{lemma4.1}.}
\label{Fig3}
\end{figure}
\end{center}

From Lemma \ref{lemma4.1}, we can prove the following proposition on an invariant set of the Winfree sphere model.

\begin{proposition}[Invariant set of the Winfree sphere model]\label{prop4.1}
Suppose that $0<\gamma<\frac{\pi}{2}$ and natural frequencies satisfy the following condition:
\[
\sin\gamma \tilde{I}(\gamma)> \frac{1}{\kappa}\max_{j\in\mathcal{N}}\|\Omega_j\|_{\mathrm{op}},
\]
and let $\{x_i\}$ be a solution of the Winfree sphere model \eqref{C-5}. Then, $\mathcal{D}_\gamma$ is an invariant set of the given system. i.e., if $\{x_i^0\}\subset \mathcal{D}_\gamma$ then $\{x_i(t)\}\subset \mathcal{D}_\gamma$ for all $t\geq0$.
\end{proposition}
\begin{proof}
Let $\{x_i^0\}\subset \mathcal{D}_\gamma$, and we assume that there exists $x_j$ such that $ \langle x_j, e\rangle=\cos \gamma$. Then we have $I(x_k)\geq I(x_j)=\tilde{I}(\gamma)$ for all $k\in\mathcal{N}$. By an assumption $\sin\gamma \tilde{I}(\gamma)> \frac{1}{\kappa}\|\Omega_j\|_{\mathrm{op}}$ and Lemma \ref{lemma4.1}, we have
\begin{align*}
\frac{d}{dt}\angle(x_j, e)&\leq \|\Omega_j\|_{\mathrm{op}}-\frac{\kappa}{N}\sin\angle(x_j, e)\sum_{k=1}^N\tilde{I}(\angle(x_k, e))
<\|\Omega_j\|_{\mathrm{op}}-\kappa\sin\gamma\tilde{I}(\gamma)<0.
\end{align*}
Since $\frac{d}{dt}\angle(x_j, e)<0$, we know that $\mathcal{D}_\gamma$ is an invariant set of the given system. 
\end{proof}

Now we study complete aggregation of the Winfree sphere model. From now on, for simplicity, we set
\[
\varphi_i(t):=\angle(x_i(t), e),\quad \varphi_i^0:=\angle(x_i^0, e),\quad \forall ~i\in\mathcal{N},~t\geq0.
\]

\begin{lemma}\label{lemma4.2}
Suppose that there exists an index $j\in \mathcal{N}$ such that $\varphi_j^0<\beta$ and for some $\gamma\in [\varphi_j^0, \beta)$ such that
\[
\|\Omega_j\|_{\mathrm{op}}<\frac{\kappa}{N}\sin\gamma\tilde{I}(\gamma),
\]
and let $\{x_i\}$ be a solution of system \eqref{C-5}. Then, we have
\[
\varphi_j(t)\leq\gamma \quad\forall t\geq0.
\]
\end{lemma}

\begin{proof}

From Lemma \ref{lemma4.1}, we have
\[
\dot{\varphi}_j\leq \|\Omega_j\|_{\mathrm{op}}-\frac{\kappa}{N}\sin\varphi_j\sum_{k=1}^N\tilde{I}(\varphi_k).
\]
Since $\tilde{I}\geq0$, we have
\[
\dot{\varphi}_j\leq \|\Omega_j\|_{\mathrm{op}}-\frac{\kappa}{N}\sin\varphi_j\tilde{I}(\varphi_j).
\]
By a simple argument of the differential inequality, if there exists $\gamma\in[\varphi_j^0, \beta)$ such that
\[
\|\Omega_j\|_{\mathrm{op}}<\frac{\kappa}{N}\sin\gamma\tilde{I}(\gamma),
\]
then $\varphi_j(t)\leq\gamma$ for all $t\geq0$. 
\end{proof}

\begin{remark}
(1) In the condition of Lemma \ref{lemma4.2}, 
\[
\|\Omega_j\|_{\mathrm{op}}<\frac{\kappa}{N}\sin\gamma\tilde{I}(\gamma)
\]
is reasonable, since for a sufficiently large coupling strength $\kappa>0$ satisfies this condition.\\

\noindent(2) If $\|\Omega_j\|_{\mathrm{op}}=0$, then the given condition always satisfies for all $\gamma\in [\varphi_j^0, \beta)$.

\end{remark}

\begin{proposition}\label{P4.2}
Suppose that there exists an index $l\in \mathcal{N}$ such that $\varphi_l^0<\beta$ and for some $\gamma\in [\varphi_l^0, \beta)$ such that
\begin{align}\label{D-6}
\max_{i\in\mathcal{N}}\|\Omega_i\|_{\mathrm{op}}<\frac{\kappa}{N}\sin\gamma\tilde{I}(\gamma).
\end{align}
We also assume that 
initial data satisfy:
\begin{align}\label{D-7}
\varphi_j^0<\pi-\arcsin\left(\frac{N\|\Omega_j\|_{\mathrm{op}}}{\kappa\tilde{I}(\gamma)}\right).
\end{align}
Let $\{x_i\}$ be a solution of system \eqref{C-5}, then we have
\[
\liminf_{t\to\infty}\varphi_j(t)\leq\arcsin\left(\frac{N\|\Omega_j\|_{\mathrm{op}}}{\kappa\tilde{I}(\gamma)}\right)\quad\forall j\in\mathcal{N}.
\]
\end{proposition}

\begin{proof}
From Lemma \ref{lemma4.2} and the given condition, we know that $\varphi_j(t)\leq \gamma$ for all $t\geq0$. We substitute this result into the inequality \eqref{D-3} to get
\begin{align*}
\dot{\varphi_j}&\leq \|\Omega_j\|_{\mathrm{op}}-\frac{\kappa}{N}\sin\varphi_j \tilde{I}(\varphi_l)
=\|\Omega_j\|_{\mathrm{op}}-\frac{\kappa}{N}\sin\varphi_j \tilde{I}(\gamma).
\end{align*}
Then $\varphi_j$ decrease when
\[
\sin\varphi_j>\frac{N\|\Omega_j\|_{\mathrm{op}}}{\kappa \tilde{I}(\gamma)}\quad\Longleftrightarrow\quad \arcsin\left(\frac{N\|\Omega_j\|_{\mathrm{op}}}{\kappa\tilde{I}(\gamma)}\right)<\varphi_j<\pi-\arcsin\left(\frac{N\|\Omega_j\|_{\mathrm{op}}}{\kappa\tilde{I}(\gamma)}\right).
\]
Here, $\arcsin\left(\frac{N\|\Omega_j\|_{\mathrm{op}}}{\kappa\tilde{I}(\gamma)}\right)$ is well-defined and less than $\frac{\pi}{2}$ from the condition \eqref{D-6}. So, we can obtain that
\[
\liminf_{t\to\infty}\varphi_j(t)\leq\arcsin\left(\frac{N\|\Omega_j\|_{\mathrm{op}}}{\kappa\tilde{I}(\gamma)}\right),
\]
and this leads to a desired result.
\end{proof}
From this proposition, we have the following theorem about transition to invariant set $\mathcal{D}_\beta$ of the Winfree sphere model. We mimic the proof of Theorem \ref{thm2.1} to prove Theorem \ref{thm4.2}. The first, second, and third steps of the proof for Theorem \ref{thm2.1} correspond to Theorem \ref{thm4.1}, Lemma \ref{lemma4.3}, and argument between Lemma \ref{lemma4.3} and Theorem \ref{thm4.2}, respectively.\\

\begin{theorem}\label{thm4.1}
Suppose that there exists an index $l\in \mathcal{N}$ such that $\varphi_l^0<\beta$ and a coupling strength $\kappa$ is larger than the critical coupling strength $\kappa_c$:
\begin{align}\label{D-8}
\kappa> \kappa_c:=\frac{N\max_{j\in\mathcal{N}}\|\Omega_j\|_{\mathrm{op}}}{\sin\beta\tilde{I}(\gamma)}
\end{align}
for some $\gamma\in [\varphi_l^0, \beta)$, and initial data satisfy \eqref{D-7}. Let $\{x_i\}$ be a solution of system \eqref{C-5}. Then we have the following relation:
\begin{align}\label{D-9}
\liminf_{t\to\infty}\varphi_j(t)\leq \arcsin\left(\frac{\kappa_c\sin\beta}{\kappa}\right)<\beta.
\end{align}
If $\varphi_j(s)<\beta$ for some $s>0$, then $\varphi_j(t)<\beta$ for all $s<t$. i.e. there exists $T_j>0$  such that $x_j(t)\in\mathcal{D}_\beta$ for all $t>T_j$. Furthermore, we can find a specific $T_j$ as follows:
\begin{align}\label{D-10}
T_j=\frac{N}{\kappa\tilde{I}(\gamma)\cos\varphi_j^0}\ln\left(\frac{\varphi_j^0-\beta_*}{\beta-\beta_*}\right).
\end{align}
\end{theorem}

\begin{proof}
If we choose such $\kappa$ satisfies \eqref{D-8} and apply Proposition \ref{P4.2}, then we can obtain \eqref{D-9} easily. Since $\frac{\kappa_c\sin\beta}{\kappa}<\sin\beta$, we have $\arcsin\left(\frac{\kappa_c\sin\beta}{\kappa}\right)<\beta$. Now, we find a specific $T_j$ which satisfyies the statement of this theorem. If $\varphi_j^0\leq \beta$, then $T_j=0$ and it is not an interesting case. So, we assume that $\varphi_j^0>\beta$ for some $j\in\mathcal{N}$. Then we have the following differential inequality:
\begin{align*}
\dot{\varphi}_j\leq\|\Omega_j\|_{\mathrm{op}}-\frac{\kappa}{N}\sin\varphi_j\tilde{I}(\gamma)
<\frac{\kappa}{N}\tilde{I}(\gamma)\left(\frac{\kappa_c}{\kappa}\sin\beta-\sin\varphi_j\right)
\end{align*}
by definition of $\kappa_c$ in \eqref{D-8}. Now we define 
\[
\beta_*=\arcsin\left(\frac{\kappa_c}{\kappa}\sin\beta\right).
\]
Then we have
\[
\dot{\varphi}_j<-\frac{\kappa}{N}\tilde{I}(\gamma)(\sin\varphi_j-\sin\beta_*)=-\frac{\kappa}{N}\tilde{I}(\gamma)(\varphi_j-\beta_*)\cos\tilde{\beta},
\]
where $\tilde{\beta}\in[\beta_*, \varphi_j]$. Since $\cos\tilde{\beta}\geq \cos\varphi_j^0$, we get
\[
\dot{\varphi}_j<-\frac{\kappa}{N}\tilde{I}(\gamma)\cos\varphi_j^0(\varphi_j-\beta_*),
\]
and this yields
\[
\frac{d}{dt}\ln(\varphi_j-\beta_*)<-\frac{\kappa\tilde{I}(\gamma)\cos\varphi_j^0}{N}.
\]
Now, we set $\tilde{T}_j=\inf\{t\geq0: \varphi_j(t)\leq\beta\}$ and we integrate the above differential inequality on $t\in[0,  \tilde{T}_j]$ to get
\[
\ln\left(\frac{\varphi_j(\tilde{T}_j)-\beta_*}{\varphi_j^0-\beta_*}\right)<-\frac{\kappa\tilde{I}(\gamma)\cos\varphi_j^0\tilde{T}_j}{N}.
\]
This implies
\[
\tilde{T}_j<-\frac{N}{\kappa\tilde{I}(\gamma)\cos\varphi_j^0}\ln\left(\frac{\beta-\beta_*}{\varphi_j^0-\beta_*}\right)=\frac{N}{\kappa\tilde{I}(\gamma)\cos\varphi_j^0}\ln\left(\frac{\varphi_j^0-\beta_*}{\beta-\beta_*}\right).
\]
Here, we used the fact $\varphi_j(\tilde{T}_j)=\beta$. Since $T_j$ defined as \eqref{D-10} satisfies $\tilde{T}_j<T_j$, we know that such $T_j$ satisfies $x_j(t)\in\mathcal{D}_\beta$ for all $t>T_j$.
\end{proof}
From Theorem \ref{thm4.1}, we have the following remark.

\begin{remark}\label{rmk4.2}
If we define the index set $\mathcal{I}:=\{j: \varphi_j^0>\beta\}$ and we define $T=\sup_{j\in\mathcal{I}}T_j$, where $T_j$ defined as \eqref{D-10}. Then $\{x_i(t)\}\subset \mathcal{D}_\beta$ for all $t>T$. 
\end{remark}

\subsection{Uniform $\ell^1$-stability and phase-locked state}
In this subsection, we study uniform $\ell^1$ stability and the existence of an equilibrium solution under suitable condition. From these two properties, we prove the phase-locked state of the Winfree sphere model with a sufficiently large coupling strength $\kappa$. From now on, for the handy notation, we define $I_c(t):=\frac{1}{N}\sum_{k=1}^NI(x_k(t))$. 

\begin{lemma}[Uniform stability]\label{lemma4.3}
Let $\{x_i\}$ and $\{\tilde{x}_i\}$ be solutions of system \eqref{C-5} with the initial data $\{x_i^0\}$ and $\{\tilde{x}_i^0\}$, respectively. If the following conditions hold a priori:
\[
\{x_i\},\{\tilde{x}_i\}\subset \mathcal{D}_\gamma
\]
for some $0<\gamma<\beta$, and positive constant $C=\cos\gamma\tilde{I}(\gamma)-\mathrm{Lip}(\tilde{I})>0$. 
Then, we have
\[
\|\mathcal{X}(t)-\tilde{\mathcal{X}}(t)\|_{\ell^1}\leq e^{-Ct}\cdot \|\mathcal{X}^0-\tilde{\mathcal{X}}^0\|_{\ell^1}\quad t\geq0,
\]
where the $\ell^1$-distance between $\mathcal{X}$ and $\tilde{\mathcal{X}}$ given as follows:
\[
\|\mathcal{X}-\tilde{\mathcal{X}}\|_{\ell^1}=\sum_{i=1}^N\|x_i-\tilde{x}_i\|.
\]
\end{lemma}

\begin{proof}
Let $\mathcal{X}:=\{x_i\}$ and $\tilde{\mathcal{X}}:=\{\tilde{x}_i\}$ be solutions to the system \eqref{C-5} with the initial data $\{x_i^0\}$ and $\{\tilde{x}_i^0\}$, respectively. From simple calculations, we have
\begin{align*}
-\frac{d}{dt}\langle x_i, \tilde{x}_i\rangle
&=-\langle x_i, \Omega_i\tilde{x}_i+\kappa \tilde{I}_c(e-\langle \tilde{x}_i, e\rangle \tilde{x}_i)\rangle-\langle \Omega_ix_i+\kappa I_c(e-\langle x_i, e\rangle x_i), \tilde{x}_i\rangle\\
&=-\kappa\Big(\tilde{I}_c(\langle x_i, e\rangle-\langle \tilde{x}_i, e\rangle\langle x_i, \tilde{x}_i\rangle)+I_c(\langle \tilde{x}_i, e\rangle-\langle x_i, e\rangle \langle x_i, \tilde{x}_i\rangle\Big).
\end{align*}
We use a relation $-\langle x_i, \tilde{x}_i\rangle=\frac{1}{2}\|x_i-\tilde{x}_i\|^2-1$ to obtain
\begin{align*}
\frac{1}{2}\frac{d}{dt}\|x_i-\tilde{x}_i\|^2&=\kappa(\langle x_i-\tilde{x}_i, e\rangle)(I_c-\tilde{I}_c)- \frac{\kappa}{2}(\langle x_i, e\rangle I_c+\langle \tilde{x}_i, e\rangle\tilde{I}_c)\|x_i-\tilde{x}_i\|^2\\
&\leq\frac{\kappa}{N}\sum_{k=1}^N\|x_i-\tilde{x}_i\|\cdot|I(x_k)-I(\tilde{x}_k)|- \frac{\kappa}{2}(\langle x_i, e\rangle I_c+\langle \tilde{x}_i, e\rangle\tilde{I}_c)\|x_i-\tilde{x}_i\|^2\\
&\leq\frac{\kappa}{N}\sum_{k=1}^N\mathrm{Lip}(I)\cdot\|x_i-\tilde{x}_i\|\cdot\|x_k-\tilde{x}_k\|- \kappa\cos\gamma \tilde{I}(\gamma) \|x_i-\tilde{x}_i\|^2,
\end{align*}
where $\mathrm{Lip}(I)$ is a Lipschitz constant of the function $I$. Finally, we get
\[
\frac{d}{dt}\|x_i-\tilde{x}_i\|\leq\frac{\kappa}{N}\sum_{k=1}^N\mathrm{Lip}(I)\cdot\|x_k-\tilde{x}_k\|- \kappa\cos\gamma \tilde{I}(\gamma) \|x_i-\tilde{x}_i\|.
\]
By definition of $\ell^1$-distance, we have
\begin{align*}
\frac{d}{dt}\|\mathcal{X}-\tilde{\mathcal{X}}\|_{\ell^1}&=\sum_{i=1}^N\frac{d}{dt}\|x_i-\tilde{x}_i\|\\
&=\sum_{i=1}^N\left(\frac{\kappa}{N}\sum_{k=1}^N\mathrm{Lip}(I)\cdot\|x_k-\tilde{x}_k\|- \kappa\cos\gamma \tilde{I}(\gamma) \|x_i-\tilde{x}_i\|\right)\\
&=-\kappa\Big(\cos\gamma\tilde{I}(\gamma)-\mathrm{Lip}(\tilde{I})\Big)\sum_{i=1}^N\|x_i-\tilde{x}_i\|.
\end{align*}
In the last inequality, we used that $\mathrm{Lip}I=\mathrm{Lip}\tilde{I}$. From the assumption $C=\cos\gamma\tilde{I}(\gamma)-\mathrm{Lip}(\tilde{I})>0$, we have
\[
\|\mathcal{X}(t)-\tilde{\mathcal{X}}(t)\|_{\ell^1}\leq e^{-Ct}\cdot \|\mathcal{X}^0-\tilde{\mathcal{X}}^0\|_{\ell^1}\quad t\geq0.
\]
This gives a desired result.
\end{proof}
\begin{remark}
Since $\cos\gamma\tilde{I}(\gamma)$ is a decreasing function of $\gamma$ on $[0, \pi/2]$ and $\cos 0\tilde{I}(0)=1$, if $\mathrm{Lip}(\tilde{I})<1$, there exists positive $\gamma^0$ which satisfies
\[
0\leq \gamma<\gamma^0\quad\Longrightarrow\quad \cos\gamma\tilde{I}(\gamma)-\mathrm{Lip}(\tilde{I})>0.
\]
\end{remark}

We consider restricted natural frequency matrices to use the result of the original Winfree model. For any $\Omega_i$, we assume that there exists a unit vector $n_i$ which perpendicular to $e$ such that $\Omega n=0$ for all $n\not\in \mathrm{span}\{n_i, e\}$. By the linear algebra, there exists $\nu_i\in\bbr$ such that
\[
\Omega_i n_i=-\nu_i e,\quad \Omega_i e=\nu_i n_i.
\]
Now we want to find a solution of
\begin{align}\label{D-20}
\Omega_ix_i+\kappa I_c(e-\langle x_i, e\rangle x_i)=0,\quad i\in\mathcal{N},
\end{align}
where $x_i\in\bbs^d$. Now we restrict the domain to $x_i\in \bbs^d\cap  \mathrm{span}\{n_i, e\}$. If there exists a solution in the restricted domain, then of course there exists a solution in the original domain $(\bbs^d)^N$. Now we substitute
\[
x_i=\cos\varphi_i e+\sin\varphi_i n_i
\]
into \eqref{D-20} to get
\[
\nu_i\cos\varphi_i n_i-\nu_i\sin\varphi_i e+\kappa I_c(e-\cos^2\varphi_i e-\cos\varphi_i\sin\varphi_i n_i)=0.
\]
Since $e\perp n_i$, we have
\[
\cos\varphi_i(\nu_i-\sin\varphi_i\kappa I_c)=0,\quad \sin\varphi_i(\nu_i- \sin\varphi_i\kappa I_c)=0.
\]
This is equivalent to 
\[
\nu_i-\kappa \tilde{I}_c\sin\varphi_i =0,\quad \forall i\in\mathcal{N},
\]
where $\tilde{I}_c=\frac{1}{N}\sum_{k=1}^N\tilde{I}(|\varphi_k|)$. Since $\tilde{I}_c>0$, if we set $\kappa \tilde{I}_c=\lambda$, we have
\[
\sin\varphi_i=\frac{\nu_i}{\lambda}.
\]
This yields
\[
\tilde{I}_c=\frac{1}{N}\sum_{k=1}^N\tilde{I}(\arcsin(|\nu_k|\lambda^{-1})).
\]
If there exists $\lambda$ such that
\begin{align}\label{D-21}
\frac{\lambda}{\kappa}-\frac{1}{N}\sum_{k=1}^N\tilde{I}(\arcsin(|\nu_k|\lambda^{-1}))=0,
\end{align}
then a solution exists. We know that the L.H.S. of \eqref{D-21} is an increasing function of
\[
\lambda\in \left[\frac{\max_{i\in\mathcal{N}}|\nu_i|}{\sin\beta}, \infty\right),
\]
since $\tilde{I}$ is decreasing and $\arcsin$ is increasing. Now, we set the R.H.S. of \eqref{D-21} as $F(\lambda)$. Then we have
\[
F\left(\frac{\max_{i\in\mathcal{N}}|\nu_i|}{\sin\beta}\right)=\frac{\max_{i\in\mathcal{N}}|\nu_i|}{\kappa\sin\beta}-\frac{1}{N}\sum_{k=1}^N\tilde{I}\left(\arcsin\left(\frac{|\nu_k|\sin\beta}{\max_{i\in\mathcal{N}}|\nu_i|}\right)\right).
\]
If we impose the following condition:
\[
\kappa> \kappa_c:=\frac{N\max_{j\in\mathcal{N}}\|\Omega_j\|_{\mathrm{op}}}{\sin\beta\tilde{I}(\gamma)},
\]
then we have
\begin{align}
\begin{aligned}\label{D-22}
F\left(\frac{\max_{i\in\mathcal{N}}|\nu_i|}{\sin\beta}\right)&=\frac{\max_{i\in\mathcal{N}}|\nu_i|}{\kappa\sin\beta}-\frac{1}{N}\sum_{k=1}^N\tilde{I}\left(\arcsin\left(\frac{|\nu_k|\sin\beta}{\max_{i\in\mathcal{N}}|\nu_i|}\right)\right)\\
&<\frac{\tilde{I}(\gamma)}{N}-\frac{1}{N}\sum_{k=1}^N\tilde{I}\left(\arcsin\left(\frac{|\nu_k|\sin\beta}{\max_{i\in\mathcal{N}}|\nu_i|}\right)\right)\\
&\leq\frac{\tilde{I}(\gamma)}{N}-\frac{1}{N}\tilde{I}\left(\arcsin\left(\frac{|\nu_i|\sin\beta}{\max_{i\in\mathcal{N}}|\nu_i|}\right)\right)=0.
\end{aligned}
\end{align}
Here we used $\max_{i\in\mathcal{N}}|\nu_i|=\max_{i\in\mathcal{N}}\|\Omega_i\|_{\mathrm{op}}$ in the second inequality. We also have
\begin{align}\label{D-23}
\lim_{\lambda\to\infty}F(\lambda)=\infty,
\end{align}
since $\lim_{x\to0}\tilde{I}(x)=0$. From relations \eqref{D-21} and \eqref{D-22}, we can apply the intermediate value theorem to obtain a solution of $F$. We can conclude that the existence of an equilibrium solution. Now we can apply Lemma \ref{lemma4.3} to obtain the following theorem.

\begin{theorem}\label{thm4.2}
Suppose that for any natural frequency matrix $\Omega_i$, there exists a unit vector $n_i$ which perpendicular to $e$ such that 
\[
\Omega n=0\quad \forall ~n\not\in \mathrm{span}\{n_i, e\},
\]
and  there exists an index $l\in \mathcal{N}$ such that $\varphi_l^0<\beta$ and for some $\gamma\in [\varphi_l^0, \beta)$. We also assume that the coupling strength $\kappa$ is larger than the critical coupling strength $\kappa_c$ defined in \eqref{D-8}, the initial data satisfy \eqref{D-7}, and the influence function $I$ satisfies 
\[
\cos\beta \tilde{I}(\beta)-\mathrm{Lip}(\tilde{I})>0.
\]
Let $\{x_i\}$ be a solution of system \eqref{C-5}. Then a solution converges to an equilibrium exponentially.
\end{theorem}

\setcounter{equation}{0}
\section{Emergent dynamics of the identical Winfree sphere model}\label{sec:5}

In this section, we study the identical Winfree sphere model given as follows:
\begin{align}
\begin{aligned}\label{E-1}
\begin{cases}
&\dot{x}_i=\Omega x_i+\displaystyle\frac{\kappa}{N}\sum_{j=1}^N(e-\langle x_i, e\rangle x_i)I(x_j),\\
&x_i(0)=x_i^0\in\bbs^d,\quad \forall i\in\mathcal{N}.
\end{cases}
\end{aligned}
\end{align}
We have solution splitting property of the identical Winfree sphere model.
\begin{proposition}[Solution splitting property]\label{prop5.1}
Let $\{x_i\}$ be a solution of system \eqref{E-1} with $\Omega\in\mathrm{Skew}_{d+1}\bbr$ which satisfying $\Omega e=0$. If we substitute $y_i=\exp(-t\Omega )x_i$, then $\{y_i\}$ is a solution of the following system:
\[
\begin{cases}
\displaystyle\dot{y}_i=\frac{\kappa}{N}(e-\langle y_i, e\rangle y_i)\sum_{j=1}^NI(y_j),\\
y_i(0)=x_i^0,\quad \forall ~i\in \mathcal{N}.
\end{cases}
\]
\end{proposition}

\begin{proof}
From direct calculations, we have
\begin{align}
\begin{aligned}\label{E-1-1}
\dot{y}_i&=\exp(-t\Omega )\dot{x}_i-\Omega \exp(-t\Omega ) x_i
=\exp(-t\Omega )\Big(\dot{x}_i-\Omega x_i\Big)\\
&=\exp(-t\Omega )\Big(\Omega x_i+\displaystyle\frac{\kappa}{N}(e-\langle x_i, e\rangle x_i)\sum_{j=1}^NI(x_j)-\Omega x_i\Big)\\
&=\frac{\kappa}{N}\Big(\exp(-t\Omega )e-\langle \exp(-t\Omega )x_i, \exp(-t\Omega )e\rangle\exp(-t\Omega )x_i\Big)\sum_{j=1}^N I(\exp(-t\Omega )y_j)\\
&=\frac{\kappa}{N}\Big(\exp(-t\Omega )e-\langle y_i, \exp(-t\Omega )e\rangle y_i\Big)\sum_{j=1}^N I(\exp(-t\Omega )y_j).
\end{aligned}
\end{align}
By $\Omega e=0$, we have
\[
\frac{d}{dt}\exp(-t\Omega)e=-\exp(-t\Omega )(\Omega e)=0
\]
and this yields
\begin{align}\label{E-1-2}
e=\exp(-0\cdot \Omega )e=\exp(-t\Omega)e,\quad \forall t\in\bbr.
\end{align}
From this fact, we also have
\begin{align}\label{E-1-3}
\langle e, \exp(-t\Omega)y\rangle=\langle \exp(t\Omega) e,y\rangle=\langle e, y\rangle
\end{align}
and this implies
\begin{align}\label{E-1-4}
I(y_j)=\tilde{I}(\angle(y_j, e))=\tilde{I}(\angle(\exp(-t\Omega)y_j, e))=I(\exp(-t\Omega)y_j).
\end{align}
If we substitute the results of \eqref{E-1-2}, \eqref{E-1-3}, and \eqref{E-1-4} into \eqref{E-1-1}, then we get
\[
\dot{y}_i=\frac{\kappa}{N}(e-\langle y_i, e\rangle y_i)\sum_{j=1}^NI(y_j).
\]
\end{proof}

\subsection{Complete synchronization}
In this subsection, we study complete synchronization of the identical Winfree sphere model.
\begin{proposition}\label{prop5.2}
Let $\{x_i\}$ be a solution of solution \eqref{E-1}. Then we have the following identity:
\begin{align}\label{E-2}
\frac{d}{dt}\|x_i-x_j\|^2=-\kappa I_c\|x_i-x_j\|^2\left(\langle e, x_i\rangle+\langle e, x_j\rangle\right).
\end{align}
Furthermore, if $x_i^0\neq x_j^0$, then the above relation can be rewritten as follows:
\[
\frac{d}{dt}\ln\left(\|x_i-x_j\|^2\right)=-\kappa I_c \left(\langle e, x_i\rangle+\langle e, x_j\rangle\right).
\]
\end{proposition}

\begin{proof}
(i) From the simple calculation, we have 
\begin{align*}
\frac{d}{dt}\|x_i-x_j\|^2&=\frac{d}{dt}\left(2-2\langle x_i, x_j\rangle\right)\\
&=-2\langle \dot{x}_i, x_j\rangle-2\langle x_i, \dot{x}_j\rangle\\
&=-2\left\langle \Omega x_i+\kappa I_c (e-\langle x_i, e\rangle x_i), x_j\right\rangle-2\left\langle x_i, \Omega x_j+\kappa I_c (e-\langle x_j, e\rangle x_j)\right\rangle\\
&=-2\kappa I_c \left(\langle e, x_j\rangle-\langle x_i, e\rangle\langle x_i, x_j\rangle+\langle x_i, e\rangle-\langle x_j , e\rangle\langle x_i, x_j\rangle\right)\\
&=-2\kappa I_c (1-\langle x_i, x_j\rangle)\left(\langle e, x_i\rangle+\langle e, x_j\rangle\right)\\
&=-\kappa I_c \|x_i-x_j\|^2\left(\langle e, x_i\rangle+\langle e, x_j\rangle\right).
\end{align*}
\noindent(ii) Suppose initial data satisfy $x_i^0\neq x_j^0$, then we know that
\[
\|x_i(t)-x_j(t)\|\neq0
\]
for all $t\geq0$. This yields
\[
\frac{d}{dt}\ln\left( \|x_i-x_j\|^2\right)=-\kappa I_c \left(\langle e, x_i\rangle+\langle e, x_j\rangle\right).
\]
\end{proof}

From Proposition \ref{prop5.2}, we can show that $\angle(x_i, x_j)$ is a decreasing function under some suitable conditions.

\begin{theorem}\label{thm5.1}
Suppose that initial data satisfy the following relation:
\begin{align}\label{E-3}
\angle (x_i^0, x_j^0)<\frac{\pi}{2}-\beta\quad\forall~ i, j\in\mathcal{N},
\end{align}
and let $\{x_i\}$ be a solution of the identical Winfree sphere model $\eqref{E-1}$. Then we have
\[
\frac{d}{dt}\|x_i-x_j\|^2=-\kappa I_c\|x_i-x_j\|^2(\langle e, x_i\rangle+\langle e, x_j\rangle)\leq 0
\]
and
\[
\angle (x_i(t), x_j(t))\leq \angle(x_i^0, x_j^0)
\]
for all $i, j\in\mathcal{N}$ and $t\geq0$.
\end{theorem}

\begin{proof}
Since $\angle(x_i^0, x_j^0)<\frac{\pi}{2}-\beta$ for all $i, j\in\mathcal{N}$, there exists $\varepsilon>0$ such that 
\[
\angle(x_i(t), x_j(t))<\frac{\pi}{2}-\beta \quad\forall~i, j\in\mathcal{N},~t\in[0, \varepsilon].
\]
Now we want to show that
\[
\angle(x_i(t), x_j(t))\leq \angle(x_i^0, x_j^0)\quad\forall~i, j\in\mathcal{N},~t\in[0, \varepsilon].
\]
To show this, we estimate of the temporal derivative of $\|x_i-x_j\|^2$. Here, we fix the time variable $0<t<\varepsilon$. We consider the following two cases.\\

\noindent$\bullet$ Case A (The case when $\min_{k\in\mathcal{N}}\angle(x_k, e)\leq\beta$): There is an index $m\in\mathcal{N}$ which satisfies
\[
\angle (x_m(t), e)=\min_{k\in\mathcal{N}}\angle(x_k(t), e).
\]
Since $\angle(\cdot, \cdot)$ is a metric on $\bbs^d$, we can apply the triangle inequality \eqref{C-5-1} as follows:
\[
\angle(x_l(t), e)\leq \angle (x_m(t), e)+\angle(x_l(t), x_m(t))<\beta+\left(\frac{\pi}{2}-\beta\right)=\frac{\pi}{2},\quad\forall~l\in\mathcal{N}.
\]
This means $\langle e, x_l(t)\rangle>0$ for all $l\in\mathcal{N}$. Since
\[
I_c(t)=\frac{\kappa}{N}\sum_{k=1}^NI(x_k(t))\geq \frac{\kappa}{N}\sum_{k=1}^NI(\beta)=0,
\]
we have
\[
-\kappa I_c(t)(\langle e, x_i(t)\rangle+\langle e, x_j(t)\rangle)\leq0.
\]
~

\noindent~$\bullet$ Case B (The case when $\min_{k\in\mathcal{N}}\angle(x_k(t), e)>\beta$): Recall that
\[
\min_{k\in\mathcal{N}}\angle(x_k(t), e)>\beta\quad\Longrightarrow \quad x_k^0\in\mathcal{D}_{\beta}^c\quad\forall~k\in\mathcal{N}.
\]
Since $\mathrm{supp}I=\mathcal{D}_{\beta}$, we have $I(x_k^0)=0$ for all $k\in \mathcal{N}$. So we have $I_c(t)=0$ and this implies
\[
-\kappa I_c(t)(\langle e, x_i(t)\rangle+\langle e, x_j(t)\rangle)=0.
\]
~

We combine the results of Case A and Case B, we have the following estimate:
\[
-\kappa I_c(t)(\langle e, x_i(t)\rangle+\langle e, x_j(t)\rangle)\leq0\quad\forall~i, j\in\mathcal{N},~t\in[0, \varepsilon].
\]
From the above relation and Proposition \ref{prop5.2}, we can obtain
\[
\angle(x_i(t), x_j(t))\leq \angle(x_i^0, x_j^0),\quad\forall~i, j\in\mathcal{N},~t\in[0, \varepsilon].
\]
Now we define $t^*\geq0$ as follows:
\[
t^*:=\sup\{t\geq0:\angle(x_i(s), x_j(s))\leq \angle(x_i^0, x_j^0)\quad\forall ~i, j\in\mathcal{N}, ~0\leq s\leq t\}.
\]
If we assume $t^*<\infty$, then it makes contradiction from the above argument. Finally, we have $t^*=\infty$ and this implies the desired result:
\[
\angle (x_j(t), x_i(t))\leq \angle (x_j^0, x_i^0)\quad\forall~ i, j\in\mathcal{N},\quad t\geq0.
\]
Furthermore, from this result, we can also obtain that the inequality
\[
-\kappa I_c(t)(\langle e, x_i(t)\rangle+\langle e, x_j(t)\rangle)\leq0
\]
holds for all $i, j\in\mathcal{N}$ and $t\geq0$.
\end{proof}

Next, to apply Barbalat's lemma, we show that the derivative of $x_i$ is a uniformly continuous function with respect to the time variable $t$.
\begin{lemma}\label{lemma5.1}
Let $\{x_i\}$ be a solution of the Winfree sphere model \eqref{C-5}. Then the derivative of $x_i$ uniformly continuous for all $i\in\mathcal{N}$.
\end{lemma}

\begin{proof}
Since only the Lipschitz continuity of the influence function $I$ is assumed, we cannot show the uniformly continuity of the derivative of $x_i$ from the uniformly boundedness of the second derivative of $x_i$. Since $I$ is a Lipschitz continuous function, we know that
\[
\mathcal{F}_i(\mathcal{X})=\Omega_ix_i+\frac{\kappa}{N}(e-\langle x_i, e\rangle x_i)\sum_{j=1}^NI(x_j)
\]
is a Lipschitz continuous function of $\mathcal{X}=\{x_i\}$. Let $K_i$ be a Lipschitz constant of $\mathcal{F}_i$. From the direct calculation, we have
\[
\left\|\frac{d}{dt}\mathcal{X}(t)\right\|^2= \sum_{i=1}^N\|\dot{x}_i\|^2=\sum_{i=1}^N\|\mathcal{F}_i(\mathcal{X})\|_F^2,
\]
and this yields the uniformly bounded of $\left\|\frac{d}{dt}\mathcal{X}(t)\right\|_F^2$, since $\mathcal{F}_i$ is continuous and the domain $(\bbs^d)^N$ is compact so the R.H.S. of the above inequality is uniformly bounded. So we have
\[
\|\mathcal{F}_i(\mathcal{X}(t_1))-\mathcal{F}_i(\mathcal{X}(t_2))\|\leq K_i\|\mathcal{X}(t_1)-\mathcal{X}(t_2)\|\leq K_i \left\|\frac{d}{dt}\mathcal{X}(t_*)\right\|\cdot|t_1-t_2|
\]
for some $t_*$ between $t_1$ and $t_2$. Therefore, we know that $\mathcal{F}_i(\mathcal{X}(t))$ is a Lipschitz continuous function of $t$. Finally, the derivative of $x_i$ is uniformly continuous for all $i\in\mathcal{N}$.
\end{proof}

\begin{remark}\label{rmk5.1}
Since the $x_i$ and its derivative are both uniformly continuous and bounded, we know that the inner product $\langle \dot{x}_i, x_j\rangle$ is also uniformly continuous for all $i, j\in\mathcal{N}$.
\end{remark}

From Theorem \ref{thm5.1} and Barbalat's lemma(Lemma \ref{lem2.2}), we can obtain the following corollary.

\begin{corollary}\label{coro5.1}
Let $\{x_i\}$ be a solution of the identical Winfree sphere model $\eqref{E-1}$. Let initial data satisfy
\[
\angle (x_j^0, x_i^0)<\frac{\pi}{2}-\beta\quad\forall~ i, j\in\mathcal{N}.
\]
Then, we have the following limit:
\[
\lim_{t\to\infty}I_c(\langle e, x_i\rangle+\langle e, x_j\rangle)\|x_i-x_j\|^2=0.
\]
\end{corollary}

\begin{proof}
From Theorem \ref{thm5.1}, we have
\[
\frac{d}{dt}\|x_i-x_j\|^2=-\kappa I_c(\langle e, x_i\rangle+\langle e, x_j\rangle)\|x_i-x_j\|^2\leq0.
\]
This implies that $\|x_i-x_j\|^2$ is bounded and decreasing, so we know that there exists a limit $\lim_{t\to\infty}\|x_i-x_j\|^2$. From Remark \ref{rmk5.1} and Barbalat's lemma, we can obtain
\[
\lim_{t\to\infty}\frac{d}{dt}\|x_i-x_j\|^2=\lim_{t\to\infty}I_c(\langle e, x_i\rangle+\langle e, x_j\rangle)\|x_i-x_j\|^2=0.
\]
\end{proof}

We are ready to prove the exponential synchronization of the identical Winfree sphere model.

\begin{theorem}[Exponential decay of the identical Winfree sphere model]\label{thm5.2}
Suppose that $0<\gamma<\beta$ and initial data satisfy the following conditions:
\[
\sin\gamma \tilde{I}(\gamma)> \frac{\|\Omega\|_{\mathrm{op}}}{\kappa},\quad \{x_i^0\}\subset \mathcal{D}_\gamma,
\]
and let $\{x_i\}$ be a solution of the identical Winfree sphere model \eqref{E-1}. Then exponential aggregation occurs as follows:
\[
\|x_i(t)-x_j(t)\|\leq\|x_i^0-x_j^0\|\exp(-\lambda t), 
\]
where $\lambda=\kappa\cos\gamma\tilde{I}(\gamma)>0$.
\end{theorem}

\begin{proof}
From Proposition \ref{prop4.1}, we know that $\mathcal{D}_\gamma$ is an invariant set of the system, so we have $\{x_i(t)\}\subset \mathcal{D}_\gamma$ for all $t\geq0$. It follows from \eqref{E-2} that
\[ 
\frac{d}{dt}\|x_i-x_j\|^2=-\kappa I_c\|x_i-x_j\|^2\left(\langle e, x_i\rangle+\langle e, x_j\rangle\right).
\]
Since $\{x_i(t)\}\subset \mathcal{D}_\gamma$, we have
\[
\langle e, x_i\rangle+\langle e, x_j\rangle\leq 2\cos\gamma
\]
and 
\[
I_c=\frac{1}{N}\sum_{k=1}^NI(x_k)\geq\frac{1}{N}\sum_{k=1}^N\tilde{I}(\gamma)=\tilde{I}(\gamma).
\]
By the above relation, we get
\[
\frac{d}{dt}\|x_i-x_j\|^2\leq -2\kappa\cos\gamma\tilde{I}(\gamma)\|x_i-x_j\|^2,
\]
and this yields
\[
\|x_i(t)-x_j(t)\|\leq \|x_i^0-x_j^0\|\exp(-\kappa\cos\gamma\tilde{I}(\gamma)t).
\]
\end{proof}

Now, we study the dynamics of the order parameter $R$ defined as follows:
\[
x_c=\frac{1}{N}\sum_{k=1}^N x_k,\quad R=\|x_c\|.
\]
If $\{x_i\}$ is a solution of system \eqref{E-1}, then we have 
\[
\dot{x}_c=\Omega x_c+\kappa I_c(e-\langle x_c, e\rangle x_c).
\]
This yields
\[
\frac{d}{dt}\|x_c\|^2=2\kappa I_c\langle x_c, e\rangle(1-\|x_c\|^2).
\]
If we assume $\angle(x_i^0, x_j^0)<\frac{\pi}{2}-\beta$ for all $i, j\in \mathcal{N}$, then from Theorem \ref{thm5.1}, we know that $I_c(\langle e, x_i\rangle+\langle e, x_j\rangle)\geq0$ for all $i, j\in \mathcal{N}$. This yields $I_c\langle x_c, e\rangle\geq0$. So we have
\[
\frac{d}{dt}\|x_c\|^2=2\kappa I_c\langle x_c, e\rangle(1-\|x_c\|^2)\geq0.
\]
From $\|x_c\|\leq 1$, we know that the order parameter is increasing and bounded above. This yields that there exists the limit $\lim_{t\to\infty}\|x_c\|$. If we assume that the limit is not $1$, then by Barbalat's lemma, we have
\[
\lim_{t\to\infty}I_c\langle x_c, e\rangle=0.
\]
From the geometric structure, if $\langle x_c, e\rangle=0$ then $I_c=0$. So we can conclude that
\[
\lim_{t\to\infty} I_c=0.
\]
Finally, we can see that a solution of \eqref{E-1} satisfies 
\[
\text{either}\quad \lim_{t\to\infty}\|x_c\|=1\quad\text{or}\quad\lim_{t\to\infty} I_c(t)=0.
\]
The first limit implies the complete aggregation. In this case, we use Barbalat's lemma to obtain the limit, so we cannot guarantee the exponential aggregation. The second limit implies that all particles only effected by the free-flow term $\dot{x}_i=\Omega x_i$. If the support of the influence function $I$ is too narrow, then all particles cannot be effected by the Winfree coupling. We can state this result as follows.

\begin{theorem}\label{thm5.3}
Suppose that initial data satisfy \eqref{E-3}, and let $\{x_i\}$ be a solution of the identical Winfree sphere model \eqref{E-1}. Then, the order parameter $\|x_c\|$ of the system is increasing, and we have a dichotomy for the long-time behaviors

\noindent(1) Complete synchronization occurs:
\[
\lim_{t\to\infty}\|x_c\|=1.
\]

\noindent(2) All particles escape the support of the influence function $I$:
\[
\lim_{t\to\infty}I_c(t)=0\quad\Longleftrightarrow \quad \liminf_{t\to\infty}\angle(e, x_i)\geq\beta.
\]
Furthermore, all particles move along the free flow after a sufficiently large time $T>0$:
\[
\dot{x}_i=\Omega x_i,\quad i\in\mathcal{N},~ t\geq T.
\]
\end{theorem}

\subsection{Equilibrium solutions on $\bbs^2$.}
In this subsection, we classify a set of equilibrium of the identical Winfree sphere model \eqref{E-1}. We want to find $\{x_i\}$ which satisfy the following condition:
\[
\Omega x_i+\frac{\kappa}{N}\sum_{j=1}^N(e-\langle x_i, e\rangle x_i)I(x_j)=0\quad\Leftrightarrow\quad
\Omega x_i+\kappa I_c (e-\langle x_i, e\rangle x_i)=0.
\]
We consider two cases $\Omega=0$ and $\Omega\neq0$.
\subsubsection{The case when $\Omega=0$.} In this case, we need to solve
\[
 I_c (e-\langle x_i, e\rangle x_i)=0.
\]
If $ I_c =0$, it follows from $I(x)\geq0$ that
\[
I(x_k)=0\quad\forall k\in\mathcal{N},\quad \text{i.e.} \quad\angle(x_k, e)\geq \beta.
\]
If $ I_c >0$, we have
\[
e-\langle x_i, e\rangle x_i=0,\quad \text{i.e.}\quad x_i=\pm e\quad\forall i\in\mathcal{N}.
\]
This is a set of bi-polar state. Finally, we can express the set of equilibriums as follows:
\[
\mathcal{E}=\mathcal{A}\sqcup\mathcal{B},
\]
where
\[
\mathcal{A}=\{(x_1, x_2, \cdots, x_N)\in (\bbs^d)^N: \angle(x_k, e)\geq\beta\quad\forall k\in\mathcal{N}\}
\]
and
\[
\mathcal{B}=\{(x_1, x_2, \cdots, x_N)\in (\bbs^d)^N: x_k=\pm e\quad\forall k\in\mathcal{N}\}.
\]

\subsubsection{The case when $\Omega\neq 0$.} Since the general case is too complicated, we only consider the case when $d=2$ with the axis of rotation is parallel or perpendicular to $e$. If there exists an equilibrium $(x_1, x_2, \cdots, x_N)\in(\bbs^2)^N$, then 
\begin{align}\label{E-10}
\Omega x_i+\kappa I_c(e-\langle x_i, e\rangle x_i)=0.
\end{align}
For any $\Omega\in\mathrm{Skew}_3(\bbr)$, there exists a unit vector $u=[u_1,u_2,u_3]^{\top}\in\bbr^3$ and a positive scalar $\lambda$, such that
\[
\Omega=\lambda\hat{u}=\lambda\begin{bmatrix}
0&-u_3&u_2\\
u_3&0&-u_1\\
-u_2&u_1&0
\end{bmatrix}.
\]
From the simple calculation, we know that 
\[
\Omega x=0\quad\Longleftrightarrow\quad x=\mu u,
\]
for some constant $\mu\in\bbr$. Recall that $e=[1, 0, 0]^\top$.\\

\noindent$\bullet$ Case A (The case when $u=\pm e$): If $u=\pm e$, we have that $\Omega u=\pm\Omega e=0$. From this fact and \eqref{E-10}, we have
\[
0=\langle \Omega x_i, \kappa I_c(e-\langle x_i, e\rangle x_i), e\rangle=\kappa I_c(1-\langle x_i, e\rangle^2).
\]
 In this case, an equilibrium solution only occurs when $I_c=0$ or $x_i=\pm e$. If $I_c=0$, then \eqref{E-10} yields $\Omega x_i=0$. This also yields $x_i=\pm e$.\\

\noindent$\bullet$ Case B: (The case when $u\perp e$): Without loss of generality, we set $u=[0, 1, 0]^\top$. Then, we can easily set
\[
e_1=e,\quad e_2=u,\quad e_3=e\times u,
\]
where $\times$ is a cross product between two vectors in $\bbr^3$. Then we can express $x_i$ and $\Omega$ as component forms
\[
x_i=[x_i^1, x_i^2, x_i^3]^\top,\quad
\Omega=\lambda\begin{bmatrix}
0&0&1\\
0&0&0\\
-1&0&0
\end{bmatrix}.
\]
Now we substitute $\Omega$ into \eqref{E-10}, we get
\[
x_i^3+\frac{\kappa I_c}{\lambda}(1-(x_i^1)^2)=0,\quad \frac{\kappa I_c}{\lambda}x_i^1x_i^2=0,\quad x_i^1-\frac{\kappa I_c}{\lambda}x_i^1x_i^3=0.
\]
If $I_c=0$, then the Winfree coupling term should be zero. This yields
\[
x_i^1=0,\quad x_i^3=0,\quad x_i^2=\pm1\quad\Leftrightarrow\quad x_i=\pm u.
\]
Since $\Omega u\neq0$, this is not an equilibrium solution. This implies $I_c\neq0$, and this yields
\[
x_i^1x_i^2=0.
\]
Now, we consider two cases.\\

\noindent$\diamond$ (The case when $x_i^1=0$): In this case, we have
\[
x_i^1=0,\quad x_i^2=\pm\sqrt{1-\left(\frac{\kappa I_c}{\lambda}\right)^2},\quad x_i^3=-\frac{\kappa I_c}{\lambda}.
\]
This case only occurs when $\left|\frac{\kappa I_c}{\lambda}\right|\leq1$.\\

\noindent$\diamond$ (The case when $x_i^2=0$): In this case, we have
\[
x_i^1=\pm\sqrt{1-\left(\frac{\lambda}{\kappa I_c}\right)^2},\quad x_i^2=0,\quad x_i^3=\frac{\lambda}{\kappa I_c}.
\]
This case only occurs when $\left|\frac{\kappa I_c}{\lambda}\right|\geq1$.\\

Now we parametrize $\mu=\frac{\kappa I_c}{\lambda}$ to make simpler notation. If $\mu\neq0$, then above cases covers two great circles:
\begin{align}\label{E-11}
\bbs^2\cap\{x_i: x_i^1=0\},\quad \bbs^2\cap\{x_i: x_i^2=0\}.
\end{align}
From the above result, we can obtain that any equilibrium must be contained in either $\bbs^2\cap\{x_i: x_i^1=0\}$ or $\bbs^2\cap\{x_i: x_i^2=0\}$.

\subsection{Constants of motion} In this section, we study constants of motion for the identical Winfree sphere model. To provide various constants of motion, we define some special type of natural frequency matrices set $\mathcal{SN}_{d}$. Constant of motion for matrix Riccati equations are introduced in \cite{Lo-0}. We also refers a constants of motion of the Lohe sphere model and Lohe hermitian sphere model in \cite{H-K-P-R} and \cite{H-P-(2)}, respectively.\\
\begin{definition}
For any matrix $\Omega\in\mathrm{Skew}_{d+1}\bbr$, we define the following set:
\[
\mathcal{A}(\Omega, e):=\{v\in\bbr^{d+1}: \Omega v=0,~ \langle v, e\rangle=0,~\|v\|=1\}.
\]
We can also define the subset of natural frequency matrices as follows:
\[
\mathcal{SN}_{d}:=\{\Omega\in\mathrm{Skew}_{d+1}\bbr: \mathcal{A}(\Omega, e)\neq\emptyset\}.
\]
\end{definition}
\begin{remark}\label{rmk5.2}
(1) For the case when $d=1$, $\mathcal{SN}_1$ only contains zero matrix. If $\Omega\in\mathrm{Skew}_{2}\bbr$, $\Omega$ can be expressed as the following form.
\[
\Omega=\nu\begin{bmatrix}
0&-1\\
1&0
\end{bmatrix}.
\]
From the simple calculation, if there exists $v=[v_1, v_2]^{\top}\in\bbs^1\subset\bbr^2$ such that $\Omega v=0$, then $\nu[-v_2,v_1]^{\top}=[0, 0]$. Since $v_1^2+v_2^2=1$, we can easily obtain that $\nu=0$ and this implies $\mathcal{SN}_1=\{O_{2\times 2}\}$.

\noindent(2) From the above result, the original Winfree model with $\nu_i\neq0$ for some $i\in\mathcal{N}$ cannot be reduced from the Winfree sphere model with $d=1$ and $\Omega_i\in\mathcal{SN}_1$. Only identical Winfree model can be reduced from this form. However, the general Winfree model can be reduced from the Winfree sphere model with $d=2$ and and $\Omega_i\in\mathcal{SN}_2$ with the following framework:
\[
x_i=\begin{bmatrix}
\cos\theta_i\\
\sin\theta_i\\
0
\end{bmatrix},\quad \Omega_i=\begin{bmatrix}
0&-\nu_i&0\\
\nu_i&0&0\\
0&0&0
\end{bmatrix}.
\]
Since $\Omega_iv=0$ for $v=[0, 0, 1]^{\top}$ and $\langle v, e\rangle=0$, we know that $\Omega_i\in\mathcal{SN}_2$.\\

\end{remark}

\begin{proposition}\label{prop5.3}
Let $\{x_i\}$ be a solution of system \eqref{C-5} with $\Omega_i\in \mathcal{SN}_d$ for all $i\in\mathcal{N}$. For a vector $v\in\mathcal{A}(\Omega_i, e)$, we have the following equality:
\[
\frac{d}{dt}\langle x_i, v\rangle=-\kappa \langle x_i, e\rangle \langle x_i, v\rangle I_c .
\]
Furthermore, if $\langle x_i^0, v\rangle\neq0$, we have
\[
\frac{d}{dt}\ln|\langle x_i, v\rangle|=-\kappa\langle x_i, e\rangle I_c .
\]
\end{proposition}

\begin{proof}
From the simple calculation, we have
\begin{align*}
\frac{d}{dt}\langle x_i, v\rangle&=\langle \dot{x}_i, v\rangle
=\langle v, \Omega_i x_i\rangle +\frac{\kappa}{N}\sum_{j=1}^N(\langle v, e\rangle -\langle x_i, e\rangle \langle v, x_i\rangle)I(x_j)=-\kappa \langle x_i, e\rangle \langle x_i, v\rangle I_c .
\end{align*}
In the last equality, we used $\Omega_i v=0$ and $\langle v, e\rangle=0$.
\end{proof}

From Proposition \ref{prop5.2}, we have the following corollary.
\begin{corollary}\label{coro5.2}
Suppose that for distinct indices $\{i_1, i_2, i_3, i_4\}\subset\mathcal{N}$, initial data satisfy
\[
x_{i_\alpha}^0\neq x_{i_\beta}^0,\quad\forall ~1\leq \alpha<\beta\leq 4,
\]
and let $\{x_i\}$ be a solution of system \eqref{E-1}. Then, the following cross-ratio
\[
\mathcal{C}_{i_1i_2i_3i_4}=\frac{\|x_{i_1}-x_{i_2}\|^2\cdot\|x_{i_3}-x_{i_4}\|^2}{\|x_{i_2}-x_{i_3}\|^2\cdot \|x_{i_4}-x_{i_1}\|^2}
\]
is well-defined and a constant of motion for system \eqref{E-1}.
\end{corollary}

\begin{proof}
If $x_{i_\alpha}^0\neq x_{i_\beta}^0$, then $x_{i_\alpha}(t)\neq x_{i_\beta}(t)$. So the cross ration $\mathcal{C}_{i_1i_2i_3i_4}$ is well-defined for any $t\geq0$.
From the simple fact
\begin{align*}
\ln\mathcal{C}_{i_1i_2i_3i_4}=\ln\left( \|x_{i_1}-x_{i_2}\|^2\right)+\ln\left( \|x_{i_3}-x_{i_4}\|^2\right)-\ln\left( \|x_{i_2}-x_{i_3}\|^2\right)-\ln\left( \|x_{i_4}-x_{i_1}\|^2\right),
\end{align*}
and Proposition \ref{prop5.2} yield the following calculation
\begin{align*}
\frac{d}{dt}\ln\mathcal{C}_{i_1i_2i_3i_4}&=-\kappa I_c \left(\langle e, x_{i_1}\rangle+\langle e, x_{i_2}\rangle\right)-\kappa I_c \left(\langle e, x_{i_3}\rangle+\langle e, x_{i_4}\rangle\right)+\kappa I_c \left(\langle e, x_{i_2}\rangle+\langle e, x_{i_3}\rangle\right)\\
&+\kappa I_c \left(\langle e, x_{i_4}\rangle+\langle e, x_{i_1}\rangle\right)=0.
\end{align*}
So, we know that $\mathcal{C}_{i_1i_2i_3i_4}$ is a constant of motion for system \eqref{E-1}.
\end{proof}

If we combine Propositions \ref{prop5.2} and \ref{prop5.3}, we can obtain the following theorem.
\begin{theorem}\label{thm5.4}
Suppose that initial data satisfy
\[
x_i^0\neq x_j^0,\quad \langle x_i^0, v\rangle\neq0,\quad \langle x_j^0,v\rangle\neq0\quad \forall~i\neq j\in\mathcal{N},
\]
and let $\{x_i\}$ be a solution of system \eqref{E-1} with $\Omega\in\mathcal{SN}_d$. Then for a vector $v\in\mathcal{A}(\Omega, e)$, we have the following constants of motion:
\[
\frac{\|x_i(t)-x_j(t)\|^2}{\langle x_i(t), v\rangle\cdot\langle x_j(t), v\rangle}=\frac{\|x_i^0-x_j^0\|^2}{\langle x_i^0, v\rangle\cdot\langle x_j^0, v\rangle}\quad t\geq0.
\]
\end{theorem}

\begin{proof}
By Propositions \ref{prop5.2} and \ref{prop5.3}, we have
\begin{align*}
\frac{d}{dt}\ln\left(\frac{\|x_i-x_j\|^2}{\langle x_i, v\rangle\cdot\langle x_j, v\rangle}\right)=-\kappa I_c (\langle e, x_i\rangle+\langle e, x_j\rangle)+\kappa I_c \langle e, x_i\rangle+\kappa I_c\langle e, x_j\rangle=0.
\end{align*}
So we know that
\[
\frac{\|x_i-x_j\|^2}{\langle x_i, v\rangle\cdot\langle x_j, v\rangle}
\]
is a constant of motion for the given system.
\end{proof}

From Theorem \ref{thm5.2}, we have exponential aggregation of the identical Winfree sphere model. If we combine this result and the constant of motion introduced in Theorem \ref{thm5.4}, then we have exponential decay of $\langle x_i, v\rangle$ as the following corollary.
\begin{corollary}\label{coro5.3}
Let $\{x_i\}$ be a solution of system \eqref{E-1} with $\Omega\in\mathcal{SN}_d$. If we assume exponential aggregation of $\{x_i\}$ as follows:
\[
\|x_i(t)-x_j(t)\|\leq\|x_i^0-x_j^0\|\exp(-\lambda t), \quad\forall~i, j\in\mathcal{N},~t\geq0,
\]
then for a vector $v\in\mathcal{A}(\Omega, e)$, $\langle x_i, v\rangle$ also decaies to zero exponentially:
\[
|\langle x_i(t), v\rangle|\leq C\exp(-\lambda t/2).
\]
\end{corollary}

\begin{proof}
By Theorem \ref{thm5.4}, we have
\[
|\langle x_i, v\rangle\cdot \langle x_j, v\rangle|\leq \left|\frac{\langle x_i^0, v\rangle\cdot \langle x_j^0, v\rangle \|x_i-x_j\|^2}{\|x_i^0-x_j^0\|^2}\right|\leq C_1\exp(-2\lambda t).
\]
From this relation, we use the triangle inequality to get
\begin{align*}
\langle x_i, v\rangle^2&=|\langle x_i, v\rangle\cdot \langle x_j, v\rangle+\langle x_i, v\rangle\cdot\langle x_i-x_j, v\rangle|\\
&\leq |\langle x_i, v\rangle\cdot \langle x_j, v\rangle|+|\langle x_i, v\rangle|\cdot\|x_i-x_j\|
\leq C_1\exp(-2\lambda t)+C_2\exp(-\lambda t).
\end{align*}
Finally, we obtain the desired result:
\[
|\langle x_i(t), v\rangle|\leq C\exp(-\lambda t/2).
\]
\end{proof}

\setcounter{equation}{0}
\section{conclusion}\label{sec:6}
In this work, we have studied the generalization of the Winfree model on the high-dimensional sphere. Previously in \cite{A-S, N-M-P, Q-R-S1, Q-R-S2}, the original Winfree model with a special pair of the sensitivity and influence functions $(S(\theta), I(\theta))=(-\sin\theta, 1+\cos\theta)$ was studied. For a further research, authors of \cite{H-P-R} considered generalized pair with suitable condition presented in Theorem \ref{thm2.1}. In this study, we considered another branch of generalization pair $(S, I)$. We restricted the support of the influence function to approximate the Dirac-delta function. Under this assumption, we studied emergent dynamics of the Winfree sphere model. We could obtain the phase-locked state with the special natural frequency matrices. When the same natural frequency matrices are given, we could obtain long-time emergent dynamics that do not present in the original Winfree model. We also provided the various constants of motion for the identical Winfree sphere model. In Section \ref{sec:5}, we only provided equilibrium solutions for the case when $\bbs^2$. Classifying all equilibrium solutions for the general $\bbs^d$ can be a nice future work. We only considered the generalization of the Winfree model on the sphere. In some way, we can consider the Winfree sphere model as a combination of the Winfree model and the Lohe sphere model (aggregation model on the sphere $\bbs^d$). From this point of view, we can also combine the Winfree model and the Lohe matrix model (aggregation model on the unitary group $\mathbb{U}(d)$) to obtain the generalization of the Winfree model on the unitary group. We leave the aforementioned questions for a future work.

\end{document}